\documentclass[11pt]{article}

\usepackage[toc,page]{appendix}
\usepackage{hyperref}
\usepackage[T1]{fontenc}
\usepackage{amsfonts}
\usepackage[english]{babel}
\usepackage{a4wide,times}

\usepackage[latin1]{inputenc}
\usepackage{amssymb}

\usepackage{graphicx}
\usepackage{subcaption}

\usepackage{epsfig}
\usepackage{amsbsy}
\usepackage{verbatim}
\usepackage{color}
\usepackage{mathrsfs}
\usepackage{makeidx}
\usepackage{amsmath}
\usepackage{amsthm}
\usepackage{xcolor}

\theoremstyle{plain}
\newtheorem{thm}{Theorem}[section]

\newtheorem{prop}[thm]{Proposition}

\theoremstyle{definition}
\newtheorem{defn}{Definition}[section]
\usepackage{dsfont}

\theoremstyle{remark}
\newtheorem{Remark}{\bf Remark}[section]
\theoremstyle{remark}

\newtheorem{com*}{\bf Comment}
\usepackage{fancyhdr}

\makeatletter
\def \newequation#1#2{
 \@definecounter{#1}
 \@namedef{the#1}{\hbox{#2}}
 \@namedef{#1}{$$\refstepcounter{#1}}
 \@namedef{end#1}{
    \eqno \csname the#1\endcsname $$\global\@ignoretrue
    }
}

\newequation{haar}{(Haar)}

\newcommand{\E}{\mathds E}

\newcommand{\R}{\mathbb R}

\usepackage[latin1]{inputenc}

\title{Strong Solutions and Quantization-Based Numerical Schemes for a Class of Non-Markovian Volatility Models}

\author{Martino Grasselli \thanks{Department of Mathematics ``Tullio Levi Civita'',
University of Padova, via Trieste 63, 35121 Padova, Italy, and De Vinci Higher Education, De Vinci Research Center, Paris, France. Email: grassell@math.unipd.it}
\and
Gilles Pag\`es\thanks{Sorbonne Universit\'e, Laboratoire de Probabilit\'e, Statistique et Mod\'elisation, Campus Pierre et Marie Curie, case
158, 4, pl. Jussieu, F-75252 Paris Cedex 5, France. Email: gilles.pages@upmc.fr. Acknowledgments: We are grateful to Abass Sagna for providing his support in the numerical illustration and we thank Julien Guyon, Benjamin Jourdain and Eckhard Platen for useful comments.
}
}

\begin{document}
\maketitle
\begin{abstract}
We investigate a class of non-Markovian processes that hold particular relevance in the realm of mathematical finance. This family encompasses path-dependent volatility models, including those pioneered by \cite{platrendek18} and, more recently, by \cite{GuyonVolMostlyPathDep2022}.
Our study unfolds in two principal phases. In the first phase, we introduce a functional quantization scheme based on an extended version of the Lamperti transformation that we propose to handle the presence of a memory term incorporated into the diffusion coefficient. 
In the second phase, we study the problem of existence and uniqueness of a strong solution for the SDEs related to the examples that motivate our study, in order to provide a theoretical basis to correctly apply the proposed numerical schemes. \end{abstract}

{\bf 2010 Mathematics Subject Classification}. 60F10, 91G99, 91B25.

{\bf Keywords}: Strong solutions, non Markovian SDEs, Functional quantization, Lamperti transform.

\section{Introduction}

In this paper, we investigate a family of non-Markovian processes that include some  models recently introduced in mathematical finance, including e.g.  \cite{platrendek18, GuyonVolMostlyPathDep2022}. These models share two key characteristics that distinguish them from traditional approaches: first, they represent a sophisticated alternative to the rough volatility frameworks, introduced by \cite{ElEuch2018}, \cite{rough2018} and subsequently elaborated upon by numerous authors in the field. While retaining the capacity to capture complex market dynamics (typically the S\&P and the VIX), they offer a distinct approach to modeling volatility, diverging from the fractional Brownian motion paradigm that characterizes rough volatility models. Second, these models  are parsimonious, requiring only one Brownian motion to describe the dynamics of both the underlying asset and its volatility process. This contrasts with typical stochastic volatility models such as \cite{Heston93}, SABR \cite{sabr2002} or Bergomi (\cite{bergomi1}) which often require the introduction of  multiple Brownian motions.

Let us begin with the problem at hand: given a probability space $(\Omega, \mathcal{F}, (\mathcal{F}_t)_{t\geq0}, \mathbb{P})$, consider the following (non-Markovian) stochastic process for $t\geq 0$:
\begin{equation}\label{initialprocess}
Y_t = y_0 + \int_0^t b\left(s,Y_s,\int_0^s g_1(u,Y_u)du,\int_0^s g_2(u,Y_u)dW_u\right)ds + \int_0^t 
a\left(s,Y_s,\int_0^s h(u,Y_u)du\right)dW_s,
\end{equation}
where $y_0\in\mathbb{R}$, $W$ is a $(\mathbb{P}, {\cal F})_t$-Brownian motion, independent of ${\cal F}_0$. We assume  that $Y_0= y_0$ is deterministic  for convenience and to alleviate the quantization procedure, but everything can be adapted  to the case where $Y_0$ is a random variable independent of $W$ (it suffices to add a vector quantization phase of $Y_0$ to the functional quantization of the SDE). For now, we maintain general conditions on the functions $b,g_1,g_2,a,h$, requiring only the existence of at least a weak solution to~\eqref{initialprocess} and assuming (uniform) ellipticity for the diffusion coefficient $a$, namely $a\left(s,y_s,\int_0^s h(u,y_u)du\right) \geq \epsilon_0$  for every $y\in {\cal C}_{[0,T]}(\mathbb{R}_+,\mathbb{R})$ \textcolor{black}{and} for some $\epsilon_0 > 0$, along with differentiability of $a$ with respect to both arguments.
\textcolor{black}{First notice that the process $Y$ is indeed non-Markovian, but it becomes Markovian if we add the three "companion" integrals containing $g_1, g_2$, and $h$ to it, so it is a non-Markovian process which admits a Markov representation up to an increase of the dimension of the state space (here four)}.
We shall prove that the family of processes we consider allows  for the application of efficient discretization techniques, such as product functional quantization, see  \cite{lushpages2002}, \cite{lushpages2002b}, \cite{pagesprintem05}, \cite{LusPag2023}). In fact,  for such  non-Markovian stochastic processes it is natural to call upon  functional quantization to directly approximate entire trajectories. This approach contrasts with the recursive marginal quantization approach, which only works with simulable discrete time Markov processes, inducing here an increase of the dimension, combined with the introduction of a time discretization scheme (Euler or other). 
We will show  that, thanks to an extension of the Lamperti transform, it is possible to reduce the general dynamics to very simple forms, consisting of a simple Brownian motion plus a drift. This reduction allows for the adaptation of highly effective functional quantization techniques (see  \cite{pagesprintem05}, \cite{LusPages06}, \cite{LusPag2023}).

We begin with an observation that will be of crucial importance in the following part of the paper.
From a functional quantization perspective, one might naively suggest the following approximation for the process~\eqref{initialprocess}: 
\begin{equation*}
d\widehat{Y}_t = b\left(t,\widehat{Y}_t,\int_0^t g_1(u,\widehat{Y}_u)du,\int_0^t g_2(u,\widehat{Y}_u)d\widehat W_u\right)dt + a\left(t,\widehat{Y}_t,\int_0^t h(u,\widehat{Y}_u)du\right)d\widehat{W}_t,
\end{equation*}
where $\widehat{W}_t$ denotes a functional quantization of the Brownian motion $W$. However, this approximation proves to be misleading as it neglects the presence of correction terms arising from the application of the Lamperti transformation, as we shall illustrate
 in the following (Markovian)  example.
Consider the  Brownian diffusion  
\begin{equation*}dY_t =b(Y_t)dt + a(Y_t)dW_t, \quad Y_0=y_0\in\mathbb{R},\quad t\in[0,T]\end{equation*}
  and assume that it admits at least one weak solution (for example, assume sub-linear growth of the coefficients: $\vert b(y)\vert +\vert a(y)\vert \leq (1+\vert y\vert ) \ \forall y\in \mathbb{R}$).
Moreover, assume differentiability and  (uniform) ellipticity for the diffusion coefficient $a$.
Let us now  introduce a new diffusion $X_t:=S(Y_t)$, where $S(y)=\int_0^y \frac{d\xi}{a(\xi)}$ is the Lamperti transform, so that the process $X$  will satisfy a new SDE whose diffusion
coefficient will be the constant equal to $1$, namely
\begin{equation*}
dX_t= \beta(X_t)dt+dW_t,\end{equation*}
with $\beta= \left( \frac{b}{a}-\frac{1}{2}a'\right)$.
Now, the strict positivity and differentiability of the  coefficient $a $ implies that the function $y\rightarrow S(y)$  is continuous and strictly increasing, so that its  inverse is also differentiable and satisfies the ODE: $\left( S^{-1}\right)'=a (S^{-1})$, namely, the Lamperti transform is a diffeomorphism. 
This provides us with a recipe: to quantize the SDE satisfied by $Y$, one should first formally quantize its Lampertized version $X$ (with a diffusion coefficient equal to one) by following the above naive approach since it involves solely the Brownian motion. Then, the corresponding ODE satisfied by the functional codewords of the quantizer for $X$ can be anti-transformed by taking the inverse of the Lamperti transform and we obtain the new ODE satisfied by the (functional codewords of the quantizers for the) initial process $Y$, which turns out to be:
\begin{equation*}
dy^i_t = (b(y^i_t) - \frac{1}{2}aa'(y^i_t) + a(y^i_t)\alpha_i'(t))dt,   \quad i=1,\dots, N,
\end{equation*}
with $y^i_0=Y_0$ if $Y_0$ is deterministic, and where $\alpha=(\alpha_1,\dots, \alpha_N)$ is an $N-$quantizer of the Brownian motion $W$, namely,  $\alpha_i: [0,T] \rightarrow \mathbb{R}$ denotes a codeword of a quantizer at level $N$ (that is, using at most $N$ elementary quantizers) for the Brownian motion, 
sharing some optimality properties to be specified later on.
Let us  underscore that the presence of the correction term $-\tfrac 12 aa'$, which was missing in the above naive approach,  is  similar to that obtained with regular diffusion, where the stochastic integration is taken in the Stratonovich sense, see  \textcolor{black}{\cite{Doss1977} } and \cite{PagSellami2011}, where they also make a 
 connection with rough path theory to show that the solutions of the quantized solutions of the ODE converge toward the solution of the SDE. 
 %(roughly speaking, we replace $dW_t$ by $\alpha_i'(t)dt$, and the SDE transforms into an ODE) that we will describe in the next subsection.
It turns out that the resulting quantizer $\widehat Y$
is a non-Voronoi quantizer (since it is defined on the explicit  Voronoi diagram of $W$), but it is nevertheless rate-optimal, at least in the setting of regular (Markovian) Brownian diffusions, see \cite{LusPages06} and \cite{LusPag2023} for more details.

Quite remarkably, it turns out that it is possible to apply the whole procedure even in the presence of a memory term in the diffusion coefficient of~\eqref{initialprocess}, which represents a non trivial extension of the results in \cite{LusPages06} and \cite{LusPag2023} and constitutes  the main contribution of the paper from a functional quantization perspective.
We will therefore proceed in successive steps. First, we will provide some background on the functional quantization for the Brownian motion. This initial step of the procedure is standard and is reported for the sake of completeness in the following subsection. Then, we will apply the Lamperti transform to the diffusion, and we obtain another diffusion where the coefficient in front of the Brownian motion is equal to one. To do this, we will apply the product functional quantization, and then the Lamperti inverse transform to find the ODE satisfied by the functional codewords of the quantizer for the original process in~\eqref{initialprocess}.

For illustrative purposes, we consider the model by \cite{GuyonVolMostlyPathDep2022}, where the diffusion coefficient does not contain memory terms, resulting in a particularly simple Lamperti transform, and the model by \cite{platrendek18}, where the Lamperti transform requires an extension due to the presence of an integral term in the diffusion coefficient.

The primary challenge in studying these models lies in  the analysis of existence and uniqueness of a strong solution to the associated stochastic differential equations, in order to correctly apply the whole procedure.  
 Recent literature has focused on the model proposed by \cite{GuyonVolMostlyPathDep2022}, with varying conclusions depending on the methodologies employed to prove the existence and uniqueness of strong solutions.
Notably, \cite{NutzValde2023} demonstrated the existence and uniqueness of a strong solution for the \cite{GuyonVolMostlyPathDep2022} model.  \cite{AndresJourdain2024}  proves the existence of strong solutions in a more general context, considering kernels that encompass the possibility of rough volatility models, by allowing power-type kernels in addition to the negative exponential kernels in the original \cite{GuyonVolMostlyPathDep2022} model.
%
%While the work of \cite{AndresJourdain2024} is noteworthy for its use of general localization methods that establish existence, it employs a technically complex approach applied to a context of general kernels.

In our work, we employ a more direct approach utilizing a classic existence-uniqueness result for \textit{path-dependent SDEs} from \cite{rogerswill2000} (Theorem 11.2, p128), which proves to be an extremely powerful and versatile tool for demonstrating the existence and uniqueness of strong solutions in our context. These results represent an independent contribution of the paper.
Our proofs are simpler compared to those in \cite{AndresJourdain2024}, even in the case of exponential kernels employed in the original version of 
\cite{GuyonVolMostlyPathDep2022}. On the other hand, rather than generalizing the kernel type, we consider more general dynamics that can include the aforementioned models with their respective properties: parsimony (a single Brownian motion in the process dynamics) and simplicity (negative exponential kernels, thus excluding power functions that would lead to additional technical difficulties, as they would result in non-Markovian stochastic Volterra equations).

The paper is structured as follows: Section \ref{section2} presents illustrative examples that underscore the significance of investigating the class of processes defined by Equation \ref{initialprocess}. Section \ref{section3} introduces the concept of product functional quantization, commencing with the Karhunen-Lo\`{e}ve expansion for Brownian motion. We then  extend  the classical Lamperti transform method, in order to reduce the original diffusion to a purely Brownian motion with an additional drift component. This extension is necessitated by the presence of an extra memory term in the diffusion coefficient of the process defined by Equation~\eqref{initialprocess}, allowing for the direct application of product functional quantization to transform the stochastic differential equation (SDE) into an ordinary differential equation (ODE).
Section \ref{section4} examines the model proposed by \cite{GuyonVolMostlyPathDep2022}. We first establish the existence and uniqueness of a strong solution for the corresponding path-dependent SDE and we derive the ODE arising from the classical Lamperti transform. In Section \ref{section5}, we apply an analogous procedure to the model presented in \cite{platrendek18}. This model requires additional attention due to the presence of a memory term in the diffusion coefficient, necessitating the extended Lamperti transform introduced in Section \ref{section3}. 
Section \ref{section7} concludes, with few provisional remarks on further developments, as numerical accelerating procedures.

\section{Motivating examples}\label{section2}

\subsection{Path-Dependent Volatility Models}

Recent years have witnessed the introduction of several models postulating that volatility is a function of the past trajectory of the underlying asset. The fundamental premise of these models is encapsulated in the asset dynamics:
\begin{equation*}
\frac{dS_t}{S_t} = \sigma(S_u, u \leq t) dW_t,
\end{equation*}
where the volatility $\sigma$ is assumed to be a functional of past asset returns and past squared returns, mediated through convolutional kernels.

To formalize this concept, we define two integral quantities:
\begin{align*}
R_{1,t} &:= \int_{-\infty}^t K_1(t-u) \frac{dS_u}{S_u} = \int_{-\infty}^t K_1(t-u) \sigma_u dW_u, \\
R_{2,t} &:= \int_{-\infty}^t K_2(t-u) \left(\frac{dS_u}{S_u}\right)^2 = \int_{-\infty}^t K_2(t-u) \sigma^2_u du,
\end{align*}
where $K_1$ and $K_2$ are convolution kernels  of exponential type\footnote{The choice of kernel type significantly influences the nature of the resulting volatility. While power-type kernels can generate rough volatilities, which are non-Markovian, empirical findings   suggest, according to \cite{GuyonVolMostlyPathDep2022},  that exponential kernels are more appropriate in practice.} that decay to zero:
\begin{equation*}
K_i(\tau) = \lambda_i e^{-\lambda_i \tau}, \quad \lambda_i > 0, \quad i = 1, 2.
\end{equation*}
This formulation leads to Markovian dynamics for $(R_{1,t}, R_{2,t})$:
\begin{align}
dR_{1,t}&= \lambda_1 \left(\sigma(R_{1,t},R_{2,t})dW_t-R_{1,t}dt\right)\label{R1dyn}\\
dR_{2,t}&= \lambda_2 \left(\sigma(R_{1,t},R_{2,t})^2-R_{2,t}\right)dt,\label{R2dyn}
\end{align}
which can be equivalently expressed as:
\begin{align}
R_{1,t}&=\int_{-\infty}^t \lambda_1e^{-\lambda_1 (t-u)}\sigma_udW_u\ = R_{1,0}e^{-\lambda_1 t}+\lambda_1\int_0^t e^{-\lambda_1 (t-u)}\sigma_udW_u\label{R1t},\\
R_{2,t}&= \int_{-\infty}^t \lambda_2e^{-\lambda_2 (t-u)}\sigma^2_udu = R_{2,0}e^{-\lambda_2 t}+ \lambda_2\int_0^t e^{-\lambda_2 (t-u)}\sigma^2_udu.\label{R2t}
\end{align}

\begin{Remark}
The incorporation of memory into volatility modeling has already been considered in the literature. Notable examples include:

\begin{itemize}
\item QARCH (Quadratic ARCH) model \cite{Sentana1995}:
\begin{equation*}
\sigma^2_t = \beta_0 + \beta_1 R_{1,t} + \beta_2 R_{2,t}.
\end{equation*}
\item Hawks process with Zumbach effect \cite{BlancDonierBouchaud2017, Zumbach2010}:
\begin{equation}
\sigma^2_t = \beta_0 + \beta_1 R^2_{1,t} + \beta_2 R_{2,t}\label{blanc}.
\end{equation}
\item Quadratic rough Heston \cite{Gat_Jus_Rosen2020}:
\begin{equation*}
\sigma^2_t = \beta_0 + \beta_1(R_{1,t} - \beta_2)^2.
\end{equation*}
%with a Mittag-Leffler kernel $K_1$.
\item Guyon and Lekeufack model \cite{GuyonVolMostlyPathDep2022}:
\begin{equation}
\sigma_t = \sigma(R_{1,t}, R_{2,t}) = \beta_0 + \beta_1 R_{1,t} + \beta_2 \sqrt{R_{2,t}}.\label{volJulien}
\end{equation}
\end{itemize}
\end{Remark}

These models can capture various stylized facts of financial markets. For instance, the leverage effect -- the negative correlation between returns and volatility -- can be accommodated in the models of \cite{Sentana1995} and \cite{GuyonVolMostlyPathDep2022} by setting $\beta_1 < 0$, or in \cite{Gat_Jus_Rosen2020} by ensuring $\beta_1 \beta_2 > 0$. Notably, these path-dependent volatility models do not require an independent source of noise, as volatility is endogenously generated by asset returns.

The dynamics of $R_{1,t}$ and $R_{2,t}$ are fully specified once the functional form of the volatility is determined. For example, the model of \cite{GuyonVolMostlyPathDep2022}  in~\eqref{volJulien} yields the following volatility dynamics:
\begin{equation*}
d\sigma_t = \left(-\beta_1 \lambda_1 R_{1,t} + \frac{\beta_2 \lambda_2}{2} \frac{\sigma_t^2 - R_{2,t}}{\sqrt{R_{2,t}}}\right) dt + \beta_1 \lambda_1 \sigma_t dW_t
\end{equation*}
Setting $\sigma_t = Y_t$ for all $ t\in[0,T]$, we can identify this with the general form:
\begin{equation*}
Y_t = y_0 + \int_0^t b\left(s, Y_s, \int_0^s g_1(u, Y_u) du, \int_0^s g_2(u, Y_u) dW_u\right) ds + \int_0^t a\left(Y_s, \int_0^s h(u, Y_u) du\right) dW_s,
\end{equation*}
where, for all $y\in{\cal C}([0,T],\mathbb{R})$,
\begin{align*}
\textcolor{black}{g_1(t,y_t) = e^{\lambda_2 t}y^2_t,\quad    g_2(t,y_t) = e^{\lambda_1 t}y_t}, \quad 
h(t,y_t) = 0,
\end{align*}
\begin{align}
b\Big(t,y_t,\int_0^t g_1(u,y_u)du,\int_0^t g_2(u,y_u)dW_u\Big)&=  -\beta_1 \lambda_1 e^{-\lambda_1 t}\left( R_{1,0}+\lambda_1\int_0^t \textcolor{black}{g_2(u,y_u)} dW_u\right)\nonumber\\
&+
\frac{\beta_2\lambda_2}{2}\frac{y_t^2-e^{-\lambda_2 t}\left( R_{2,0}+ \lambda_2\int_0^t \textcolor{black}{g_1(u,y_u)}\textcolor{black}{du}\right)}{\sqrt{e^{-\lambda_2 t}\left( R_{2,0}+ \lambda_2\int_0^t \textcolor{black}{g_1(u,y_u)}\textcolor{black}{du}\right)}}\label{bJulien}\\
a\Big(t,y_t,\int_0^t h(u,y_u)du\Big)&= \beta_1\lambda_1 y_t\label{aJulien}.
\end{align}

%\begin{Remark}
%The vol of vol is constant but the drift is rich enough to get price-path-dependence of volatility dynamics, thus leadind to the strong Zumbach effect. Moreover, the non negativity of volatility is guaranteed if $\lambda_2<2\lambda_1$ and thanks to the mean reverting feature, the model is capabe to generate jump-like spikes in the vol without the need to introduce a jump process!
%On the other hand, the model is not analytically tractable insofar pricing and calibration should be performed by MC simulation.
%\end{Remark}

\subsection{The Model of \cite{platrendek18}}

\textcolor{black}{We first briefly recall the pricing framework based on the Benchmark Approach introduced in \cite{bookplaten10}. Within this theory, the best-performing strictly positive portfolio is selected as the natural benchmark for asset allocation and, simultaneously, as the natural num\'eraire for pricing. This portfolio, known as the Growth Optimal Portfolio (GOP), maximizes the expected growth rate or, equivalently, the expected logarithmic utility. Moreover, it coincides with the num\'eraire portfolio (NP), in the sense that any nonnegative self-financing portfolio, when expressed in units of the NP, makes up a supermartingale.
This key property gives rise to a natural pricing rule under the real-world probability measure, yielding the minimal replicating price. Remarkably, under the benchmark approach, the existence of an equivalent risk-neutral probability measure is not required. This feature gives to the modeler  considerably greater flexibility compared to the traditional risk-neutral framework.
When taking the GOP as num\'eraire, the price at time $t$ of a contingent claim delivering a payoff 
$\Psi_T$ at maturity $T$ is given by
\begin{eqnarray}
Price_t(\Psi_T)&=& S_t \mathbb E_t^{\mathbb{P}} \left[\frac{\Psi_T}{S_T} \right].\label{benchmarkprice}
\end{eqnarray}
It should be noted that, contrary to what happens with risk-neutral pricing, when using the benchmark approach, the price of a zero-coupon bond (ZCB) involves a non-trivial average due to the presence of the inverse of the GOP in the pricing formula,  even in the case where the interest rate is constant.  Indeed, the price at time $t\leq T$ of a zero-coupon bond maturing at $T$ is expressed as the conditional expected value, under the real-world measure $\mathbb{P}$, of the unit payoff divided by the GOP, namely
\begin{eqnarray}
Price_t(ZCB_T)&=& S_t \mathbb E_t^{\mathbb{P}} \left[S^{-1}_T \right].\label{benchmarkZCB}\end{eqnarray}
The first step  consists in specifying the dynamics of the growth optimal portfolio. In practice, a well-diversified total return index serves as a suitable proxy for the GOP. To make the model directly applicable and comparable with existing index models, \cite{platrendek18} adopt the dynamics of the S\&P 500 index, which represents one of the best-studied and most diversified equity benchmarks. By modeling the market activity as a linear function of the square of the derivative of the moving average of a proxy for the underlying Brownian motion, they obtain the model that we apply here to the S\&P 500.
}

Consider the following stochastic differential equation:
\begin{equation}
\frac{dS_t}{S_t} =  r_t dt + \sqrt{X_t} (\sqrt{X_t} dt + dW_t) , \quad S_0 > 0,\label{modelPlaten}
\end{equation}
where $X_t = M_t / Y_t$, and the processes $Y_t$ and $M_t$ are defined as follows:
\begin{align}
  dY_t  &= (\alpha- \beta Y_t) M_t dt  + \sigma \sqrt{M_t Y_t}  dW_t, \qquad Y_0=y_0>0,\label{Platen1}\\
M_t & = \xi (\lambda^2 (2 \sqrt{Y_t}-Z_t)^2 + \eta  )\nonumber\\
 Z_t &= 2 \lambda \int_0^t e^{-\lambda(t-s)} \sqrt{Y_s} ds.\nonumber
\end{align}

Here, we assume $\alpha \geq \sigma^2/2$ (a Feller-like condition, see Section \ref{section5}), $\eta, \xi, \lambda \geq 0$, and $\beta > 0$.
Equation~\eqref{Platen1} represents a slight and natural extension of the evolution of the inverse of the volatility for  the Growth Optimal Portfolio (GOP) $S$ introduced by \cite{platrendek18}, which originally corresponds to:
\begin{equation}
dY_t = (1 - Y_t) M_t dt + \sqrt{M_t Y_t} dW_t, \quad Y_0 = y_0 > 0. \label{Platen2}
\end{equation}
Intuitively, the process $X$ describes the market price of risk and is related to the volatility of the GOP, in accordance with the Benchmark Approach. The process $M$ is associated with market activity and is assumed to be a function of the volatility factor $Y$ itself.
In the special case where $M$ is constant, the GOP dynamics correspond to the so-called "3/2" stochastic volatility model, where the volatility factor is the inverse of a square root process $Y$. However, a crucial distinction from the standard "3/2" model is the presence of only one Brownian motion, implying perfect correlation between the volatility and the noise driving the asset price.
Notably, the market activity process $M$ incorporates the past trajectory of the volatility factor $Y$ through its regular integral $Z$, rendering the framework non-Markovian.

From the general model presented in equation~\eqref{initialprocess}, we can identify the following components: $g_2 = 0$ and, for all $y\in{\cal C}([0,T],\mathbb{R})$,
{\small 
\begin{align}
g_1(t,y_t) = h(t,y_t) = \sqrt{y_t} e^{\lambda t}, \label{hPlaten}
\end{align}
\begin{align}
b\left(t,y_t,\int_0^t g_1(u,y_u)du,\int_0^t g_2(u,y_u)dW_u\right) &= \xi (\alpha - \beta y_t) \left[ 4\lambda^2 \left( \sqrt{y_t} - \lambda e^{-\lambda t} \int_0^t g_1(u,y_u)du \right)^2 + \eta \right], \label{bPlaten}\\
a\left(t,y_t, \int_0^t h(u,y_u)du\right) &= \sigma \sqrt{\xi y_t \left[ 4\lambda^2 \left( \sqrt{y_t} - \lambda e^{-\lambda t} \int_0^t h(u,y_u)du \right)^2 + \eta \right]}.\label{aPlaten}
\end{align}
}
Note that, in contrast to equation~\eqref{aJulien}, the diffusion coefficient in this model includes a (locally) deterministic integral term.

\section{Functional Quantization via Lamperti Transform}\label{section3}

We initiate our analysis by employing the functional quantization approach, a natural extension of optimal vector quantization for random vectors to stochastic processes. This methodology, extensively studied since the late 1940s in signal processing and information theory, aims at providing  an optimal spatial discretization of a random $\mathbb{R}^d-$valued  signal $X$ with distribution $\mathbb{P}_X$ by a random vector taking at most $N$ values $x_1,\ldots,x_N$, termed as \textit{elementary quantizers} (or codewords).
Then, instead of transmitting the complete signal $X(\omega)$ itself, one first selects the closest $x_i$ in the quantizer set and transmits its (binary or Gray coded) label $i$. After reception, a proxy $\widehat X(\omega)$ of $X(\omega)$ is reconstructed using the codebook correspondence $i\rightarrow x_i$.

For a given $N \in \mathbb{N}$, an $N$-tuple of elementary quantizers $(x_1,\ldots,x_N)$ is $L^r-$optimal, $r>0$, if it minimizes over $(\mathbb{R}^d)^N$ the $L^r-$mean  quantization error:
\begin{equation*}
\|X-\hat{X}\|_r = \min_{(y_1,\ldots,y_N) \in (\mathbb{R}^d)^N} \mathbb{E}\left[\min_{1 \leq i \leq N} \|X - y_i\|^r\right]^{1/r}
\end{equation*}
induced by replacing $X$ by $\hat X$.
Typically $r$ is fixed to be equal to 2, leading to a quadratic quantization error. In $d$ dimensions and if $X\in L^r(\mathbb{R}^d)$, the minimal $L^r-$quantization error converges to zero at a rate of $N^{-\frac{1}{d}}$ as $N \to \infty$, according to the so-called Zador
theorem, \textcolor{black}{see \cite{graflush00}}.
Several stochastic optimization procedures based on simulation have been developed to compute these optimal quantizers (see, among others, \cite{pagesprintem03}). For a comprehensive exposition of mathematical aspects of quantization in finite dimensions, we refer to \cite{graflush00} and \cite{LusPag2023}.
In the early 1990s, the field of Numerical Probability witnessed the introduction of optimal quantization to devising quadrature integration formulae with respect to the distribution $\mathbb{P}_X$ on $\mathbb{R}^d$. This method leverages the principle that $\mathbb{E}[ F(X)] \approx \mathbb{E} [F(\widehat X)]$ when $N$ is sufficiently large. 
Subsequently, optimal quantization found applications in the development of tree methods, aimed at solving multi-dimensional non-linear problems that involve the computation of numerous conditional expectations. These applications span diverse areas in computational finance, including American option pricing, non-linear filtering for stochastic volatility models, and portfolio optimization. \textcolor{black}{For a (non-exhaustive) review of its applications in computational finance, one may refer e.g. to \cite{pagesphamprint2004}, \cite{corlaypages15}, \cite{PagSagna15}, \cite{AmerQuan1}, \cite{CalFioGraRisk2}, \cite{BarrQuanSagna17}, \cite{FiorinSagna}, \cite{CalFioGra2019}, \cite{CalFioGra2}.}

Recent extensions of optimal quantization to stochastic processes, viewed as random variables taking values in their path-space, have led to significant theoretical developments (see \cite{pages2000}, \cite{Fehringer2001}, \cite{lushpages2002}, \cite{lushpages2002b}, \cite{Dereich2008a}, \cite{Dereich2008b}, \cite{pagesprintem05}, \cite{LusPag2023}). For Gaussian processes, this functional quantization can be interpreted as a discretization of the path-space, typically the Hilbert space $L^2_T:=L^2_{\mathbb{R}}([0,T],dt)$, endowed with the norm defined by $\vert f\vert_{L^2_T}=(\int_0^T f^2(t)dt)^{1/2}$. In the quadratic case,
the upper bound for  the rate of convergence of the quantization error is derived by some Hilbertian optimization methods and
the lower bound using a connection with Shannon entropy. Under some regular variation
assumptions on the ordered eigenvalues of the covariance operator of the process, the
asymptotic rates of these lower and upper bounds coincide and hence provide the exact
rate. For some  Gaussian processes, like  as example  Fractional Brownian motion and Ornstein-Uhlenbeck processes, it turns out that the rate of convergence of the quantization error is $O((\log n)^{-\mu})$, where $\mu$ is the H{\"o}lder regularity of the map $t\rightarrow X_t$ from $[0,T]$ into $L^2(\Omega, \mathbb{P})$.  In particular, for the Brownian motion $W$, for which the Karhunen-Lo\`{e}ve eigenbasis of the covariance operator is explicitly known, this result can be refined into a sharp rate of convergence $c_W (\log n)^{-1/2}$, with an explicit real constant $c_W$ that is close to the optimal one, \textcolor{black}{see \cite{LusPag2023}}. This approach can be applied to the computation of the expectation $\mathbb{E}[F (X)]$
where $X$ is a Brownian diffusion (with explicit coefficients) and $F$ is an additive (integral) functional defined on $L^2_T$ by $\xi\rightarrow F(\xi) := \int_0^T f (t, \xi(t)) dt$. Then, the quadrature formulae involving these $N-$quantizers make up an efficient
deterministic alternative to Monte Carlo simulation for the computation of $\mathbb{E}[ F (X)]$.

\subsection{Step 1: Functional Product Quantization of the Brownian Motion}\label{subsec:FPQW}

The two main families of rate optimal quantizers of the Brownian
motion are the product optimal quantizers and the true optimal quantizers: both of them are  based on the Karhunen-Lo\`{e}ve expansion, which is explicit for the Brownian motion (as well as for other Gaussian processes like e.g. the Bronwian bridge or the O-U processes).
The true optimal functional quantization, see \cite{LusPag2023}, adopts a genuinely global approach to the approximation of the underlying stochastic process and achieves slightly better performance than componentwise (optimal) product functional quantizers. However, this gain comes at the cost of a substantially higher computational burden, which may render the method less tractable in practice. As a result, product functional quantization--based on a component-by-component construction--provides an effective and well-balanced compromise between accuracy and computational feasibility for practical applications.
We begin by performing a functional quantization for the Brownian motion $W$, which we identify with $\chi^{(N)}$, where $N$ denotes the number of trajectories. For an extensive treatment of the optimal quantization of Gaussian processes, including closed-form expressions for the representation of Brownian motion in terms of an explicit Karhunen-Lo\`{e}ve basis, we refer to the works of \cite{pagesprintem03}, \cite{pagesprintem05} and \cite{LusPag2023}. Notably, \cite{pagesprintem05} employ these expansions to price path-dependent derivatives based on assets following a stochastic Heston volatility model.

\begin{Remark}
The functional quantization projects the Brownian motion into a finite-dimensional space of functions for which the stochastic integral can be defined in the usual Lebesgue-Stieltjes  sense, allowing for pathwise reasoning for each trajectory up to an arbitrary (fixed) time horizon $T$.
\end{Remark}

For a fixed trajectory $\omega \in \Omega$, the Brownian motion $W$ can be represented as:
\begin{equation}
W_t(\omega) = \sum_{\ell \geq 1} \sqrt{\lambda_{\ell}} \xi_{\ell}(\omega) e_{\ell}(t),\label{KL expansion}
\end{equation}
where, for every $\ell\geq 1$,
\begin{equation*}
\lambda_{\ell} = \left(\frac{T}{\pi(\ell - \frac{1}{2})}\right)^2,  \quad 
e_{\ell}(t) = \sqrt{\frac{2}{T}} \sin\left(\pi\left(\ell - \frac{1}{2}\right)\frac{t}{T}\right) = \sqrt{\frac{2}{T}} \sin\left(\frac{t}{\sqrt{\lambda_{\ell}}}\right),
\end{equation*}
and
\begin{equation*}
\xi_{\ell} =\sqrt{\frac{2}{T}}
\int_0^T W_t \sin \left(\frac{t}{\sqrt{\lambda_{\ell}}}\right)\frac{dt}{\sqrt{\lambda_{\ell}}}
\end{equation*}
are i.i.d., $N (0; 1)-$distributed random variables. 
Here, the sequence $(e_{\ell})_{\ell\geq 1}$ is an orthonormal basis of $L^2_T$ and the system $(\lambda_\ell, e_{\ell})_{\ell\geq 1}$ can be characterized as the eigensystem of the (symmetric positive trace class) covariance operator of $ f\rightarrow (t \rightarrow \int_0^T
 (s \wedge t) f (s)ds) $.  The Gaussian sequence $(\xi_{\ell})_{\ell \geq 1}$ is pairwise uncorrelated, namely it is an orthogonal standard Gaussian basis carrying the randomness. Intuitively, the Karhunen-Lo\`{e}ve expansion
of $W$ plays the role of PCA of the process: it is the fastest way to exhaust
the variance of $W$ among all expansions on an orthonormal basis, and it 
combines both orthonormality of the  basis $(e_\ell )_{\ell \geq 1}$ and the mutual independence of its coordinates $\xi_\ell$.

The \textit{product functional quantization} of the Brownian motion $\widehat{W}$, using at most $N$ elementary quantizers, is obtained as follows: for every $\ell\geq 1$, 
let $\Gamma^\ell=\{x_i^{(N_\ell)},i=1\dots,N_\ell\}$ some quadratic (or $L^2$-) optimal \textcolor{black}{quantization grids at level $N_\ell$ of the ${\cal N}(0,1)$ distribution}\footnote{It is possible to find explicitly the optimal $N_{\ell}-$quantizers for any $\ell$ (see \cite{pagesprintem05} for further details). In fact, these quantizers and their weights are available and can be downloaded freely on the website $www.quantize.maths-fi.com$.}, and let $\widehat\xi_\ell^{(N_\ell)}=\hbox{Proj}_{\Gamma^\ell}(\xi_\ell)$, where $\hbox{Proj}_{\Gamma^\ell}$ is the nearest neighbor projection from $\mathbb{R}^d$ to $\Gamma^\ell$. Then
\begin{equation*}
\widehat W_t(\omega) = \sum_{\ell \geq 1} \sqrt{\lambda_{\ell}} \widehat{\xi}^{(N_\ell)}_{\ell}(\omega) e_{\ell}(t),
\end{equation*}
where 
%$\widehat{\xi}_{\ell}:=\widehat{\xi}_{\ell}^{W^{(N_\ell})}=\hbox{Proj}_{W^{(N_\ell)}}(\xi_\ell)$ 
%is an optimal $N_\ell-$quantization of $\xi_\ell$ and 
$N_1\times \cdots \times N_\ell\leq N, N_1,\dots, N_\ell\geq 1$, 
so that for large enough $\ell$, $N_\ell=1$ and $\widehat{\xi}^{(N_\ell)}_{\ell}=0$ (that is the optimal $1-$quantization of a centered random variable), which makes the above series
a finite sum.
%\begin{equation}\label{dhatWexpansion}
%d\widehat{W}_t(\omega) = \sum_{\ell \geq 1} \sqrt{\frac{2}{T}} \cos\left(\frac{t}{\sqrt{\lambda_{\ell}}}\right) \widehat{\xi}_{\ell}(\omega) dt, \quad \ell \geq 1.
%\end{equation}
%
In practice, we fix the number of trajectories $N$ and truncate the summation up to $d_N$, called the $length$ of the product quantization, 
% terms\footnote{This corresponds to setting $N_{\ell} = 1$ for $\ell \geq d_N + 1$, so that $\widehat{\xi}_{\ell}^{W^{(N_\ell})} = \mathbb{E}[\xi_{\ell}^{N_{\ell}}] = 0$, and the sum is truncated at $d_N$.}, 
so that the quantizer
\begin{equation}
\widehat W_t = \sum_{\ell =1}^{d_N} \sqrt{\lambda_{\ell}} \widehat{\xi}^{(N_\ell)}_{\ell} e_{\ell}(t)
\label{quantifW}\end{equation}
takes $\prod_{\ell=1}^{d_N} N_\ell \leq N$ values.
Now, let us denote with  $(\widehat{\xi}_{i_\ell}^{(N_{\ell})}, \pi_{i_\ell}^{(N_{\ell})})_{1\leq i_\ell \leq N_\ell}$
 the optimal  $N_\ell-$quantizer of \textcolor{black}{the one-dimensional Normal distribution (with the corresponding weight $\pi_{i_\ell}^{(N_{\ell})}$)}, which is unique and explicitly known (through the c.d.f. of the Normal distribution), and let us introduce the multi-index 
 $\underline i := (i_1,\dots , i_{d_N})\in \prod_{\ell =1}^{d_N}\{ 1,\dots , N_\ell\}$. Then, the quantizer  used for $W$ (at level $N$) is made by the functional codewords
\begin{equation}\label{quantizerW}
\chi_{\underline i}^{(N)}(t) = \sum_{\ell = 1}^{d_N} \sqrt{\lambda_{\ell}}  e_{\ell}(t) \pi_{i_\ell}^{(N_{\ell})}, \quad 1\leq i_\ell \leq N_\ell, \prod_{\ell=1}^{d_N} N_\ell \leq N.
\end{equation}
The weight of the optimal quantizer is given by 
\begin{equation}
\mathbb{P}(\widehat W=\chi_{\underline i}^{(N)}) =\pi^{(N)}_{\underline i}= \prod_{\ell = 1}^{d_N} \mathbb{P}\left(\widehat\xi_\ell^{(N_\ell)}=x^{(N_\ell)}_{i_\ell}\right)=\prod_{\ell = 1}^{d_N}\pi_{i_\ell}^{(N_\ell)},
\end{equation}
where we used that the $\xi_\ell$ are independent in the K-L expansion~\eqref{KL expansion}.
 %Intuitively, the quantizer is chosen as a Cartesian product of grids of the one-dimensional standard Gaussian random variables.
The functional product quantization  $\widehat W$ can then be written as 
%through the codeword $\alpha: [0,T] \rightarrow \mathbb{R}$ as:
\begin{equation*}
\widehat{W}_t = \sum_{\underline i} \chi_{\underline i}^{(N)}(t) \mathds{1}_{ \{W \in C_{\underline i}(\Gamma^{(N)})\} },
\end{equation*}
where $\Gamma^{(N)}:=\{\chi_{\underline i}^{(N)}, \underline i \in \prod_{\ell =1}^{d_N}\{ 1,\dots , N_\ell\} \}$ and $C_{\underline i}(\Gamma^{(N)})=\prod_{\ell=1}^{d_N}C_{i_\ell}(\Gamma^\ell)$ is \textcolor{black}{the Voronoi cell of $\chi_{\underline i}^{(N)}$, that is, the set of all points that are closer to  $\chi_{\underline i}^{(N)}$
 than to any other quantization point, see see \cite{graflush00}.}
  Thus, the expected value of the additive integral functional of the Brownian motion $F(W)$ can be approximated by the following cubature formula
\begin{equation*}
\mathbb{E}F(W)\approx  \sum_{\underline i}\pi^{(N)}_{\underline i}
F(\chi_{\underline i}^{(N)}).
\end{equation*}

Finally, from~\eqref{quantifW} we can 
formally write
\begin{align*}\label{dhatWNexpansion}
d\widehat{W}_t(\omega) &=
\sum_{\ell =1}^{d_N} \sqrt{\lambda_{\ell}} \widehat{\xi}^{(N_\ell)}_{\ell} \textcolor{black}{(\omega)}d(e_{\ell}(t))\\
&= \sum_{\ell = 1}^{d_N} \sqrt{\frac{2}{T}} \cos\left(\frac{t}{\sqrt{\lambda_{\ell}}}\right) \widehat{\xi}_{\ell}^{(N_{\ell})}(\omega) dt, \quad \ell \geq 1\\
&=\sum_{\underline i} \left(
  \chi_{\underline i}^{(N)}\right)'(t) \mathds{1}_{ \{\widehat W\textcolor{black}{(\omega)}=\chi_{\underline i}^{(N)}\}}\textcolor{black}{dt}.
  \end{align*}
In the following, with a slight abuse of notation,  we will replace the last expression by $d\widehat{W}_t =\alpha'(t)dt$ in order to emphasize that the codeword of a quantizer at
level N for the Brownian motion satisfies an ODE.

\subsection{Step 2: Reduction to a Brownian Motion Plus a Drift}

In this subsection, we address the problem of eliminating the diffusion coefficient through the Lamperti transform that we are going to extend in order to allow for the integral term incorporated into the diffusion coefficient of the process \eqref{initialprocess}.
We first introduce the following integral processes:
\begin{align*}
\tilde{Y}_t^h := \int_0^t h(u,Y_u)du,\quad 
\tilde{Y}_t^{g_1} := \int_0^t g_1(u,Y_u)du, \quad
\tilde{Y}_t^{g_2} := \int_0^t g_2(u,Y_u)dW_u.
\end{align*}
We assume that:
\begin{itemize}
\item     the diffusion coefficient $a$ in~\eqref{initialprocess} is $C^1$ with respect to all arguments and (uniformly) elliptic,  namely, there exists $\epsilon_0 > 0$ such that, for every  $(\xi,\xi',\xi'')\!\in \R^3$, 
$
a(\xi,\xi',\xi'') \geq \epsilon_0.
$
\item the function $g_2(t,y)$ is $C^{1,2}$ (continously differentiable in time and twice in space).
\end{itemize}
\begin{defn}
The Lamperti transform associated with the diffusion coefficient  $a(t,y,\tilde{y}^h)$ is defined as:
\begin{equation}
\forall t\geq0, \forall y,\tilde y^h \in \mathbb{R}, \ S(t,y,\tilde{y}^h) = \int_{0}^y \frac{d\xi}{a(t,\xi, \tilde{y}^h)}.\label{lampertiytildey}
\end{equation}
\end{defn}

Under the above assumption, $S$ is well-defined and twice differentiable with respect to $y$ and once with respect to $\tilde y$ and $t$. Its partial derivatives are given by:
\begin{align*}
\frac{\partial S}{\partial y}(t,y,\tilde{y}^h) = \frac{1}{a(t,y, \tilde{y}^h)},\quad
\frac{\partial^2 S}{\partial y^2}(t,y,\tilde{y}^h) = -\frac{\frac{\partial a}{\partial y}(t,y, \tilde{y}^h)}{a^2(t,y, \tilde{y}^h)}, 
\end{align*}

\begin{align*}
\frac{\partial S}{\partial \tilde{y}^h}(t,y,\tilde{y}^h) = -\int_0^y \frac{\frac{\partial a}{\partial \tilde{y}^h}(t,\xi, \tilde{y}^h)}{a^2(t,\xi, \tilde{y}^h)}d\xi,\quad
\frac{\partial S}{\partial t}(t,y,\tilde{y}^h) = -\int_0^y \frac{\frac{\partial a}{\partial t}(t,\xi, \tilde{y}^h)}{a^2(t,\xi, \tilde{y}^h)}d\xi.
\end{align*}

Define $X_t = S(t,Y_t,\tilde{Y}_t^h)$. 
Then:
\begin{align*}
dX_t &= \left(\frac{b(t,Y_t,\tilde{Y}_t^{g_1},\tilde{Y}_t^{g_2})}{a(t,Y_t, \tilde{Y}^h_t)}
+ \frac{\partial S}{\partial t}(t,Y_t,\tilde{Y}^h_t)+ \frac{\partial S}{\partial \tilde{y}^h}(t,Y_t,\tilde{Y}^h_t)h(t,Y_t) - \frac{1}{2}\frac{\partial a}{\partial y}(t,Y_t, \tilde{Y}^h_t)\right)dt + dW_t,
\end{align*}
where we used the fact that $\tilde{Y}^h$ has finite variation.
The strict positivity of $a(t,y,\tilde{y}^h)$ implies that for every $\tilde{y}^h \in \mathbb{R}$ the function $y \mapsto S(t,y,\tilde{y}^h)$ is continuous and strictly increasing. Its continuous inverse function, denoted by $S^{-1}_y(t,.,\tilde{y}_t^h): x\rightarrow  S^{-1}_y(t,x,\tilde{y}_t^h)$, is differentiable with (positive) derivative satisfying
\begin{equation*}
\frac{\partial S^{-1}_y}{\partial x}(t,x,\tilde{y}^h) = a(t,S^{-1}_y(t,x,\tilde{y}^h), \tilde{y}^h).
\end{equation*}

The differential of $X$ can then be expressed as:
\begin{equation}\label{ODEX}
dX_t = \beta(t,S^{-1}_y(t,X_t,\tilde{Y}_t^h),\tilde{Y}_t^{g_1},\tilde{Y}_t^{g_2})dt + dW_t,
\end{equation}
where
 \begin{align*}
\beta(t,S^{-1}_y(t,X_t,\tilde{Y}_t^h),\tilde{Y}_t^{g_1},\tilde{Y}_t^{g_2}) &= \frac{b(t,S^{-1}_y(t,X_t,\tilde{Y}_t^h),\tilde{Y}_t^{g_1},\tilde{Y}_t^{g_2})}{a(t,S^{-1}_y(t,X_t,\tilde{Y}_t^h), \tilde{Y}^h_t)} 
- \frac{1}{2}\frac{\partial a}{\partial y}(t,S^{-1}_y(t,X_t,\tilde{Y}_t^h), \tilde{Y}^h_t)
 \\
&\quad + \frac{\partial S}{\partial \tilde{y}^h}(t,S^{-1}_y(t,X_t,\tilde{Y}_t^h),\tilde{Y}^h_t)h(t,S^{-1}_y(t,X_t,\tilde{Y}_t^h))
 \\
&\quad + \frac{\partial S}{\partial t}(t,S^{-1}_y(t,X_t,\tilde{Y}_t^h),\tilde{Y}^h_t).
\end{align*}
At this stage,  direct quantization of the SDE governing $X$ via the procedure outlined in Step 1 is not sufficient,   due to the presence of a Brownian integral $ \tilde{Y}_t^{g_2}$ in the drift term, which requires to be quantized as well.
The subsequent subsection delineates a methodology for the joint quantization of all Brownian terms, facilitating the quantization of~\eqref{initialprocess}. The approach aligns with the procedure established by \cite{LusPages06} and subsequently refined by \cite{LusPag2023}  through a Lamperti transform, while  for the Brownian term $ \tilde{Y}_t^{g_2}$,  it leverages  the rough path approach  developed in \cite{PagSellami2011}, in which they  show the convergence of  quantized solutions of the ODE  toward the  Brownian integral, where  stochastic integration has to be taken  in the Stratonovich sense.

\subsection{Step 3: Reduction to a System of Ordinary Differential Equations}

Having established the requisite components, we are now well-positioned to prove our principal result on functional quantization, which elucidates the ordinary differential equation governing the functional codewords of the quantizer for the original process~\eqref{initialprocess}, thereby extending the findings of \cite{LusPages06} and \cite{LusPag2023} to encompass scenarios where the diffusion coefficient incorporates a finite variation integral term.
%This result is valid under certain regularity conditions on the drift $\beta$ (e.g., standard Lipschitz assumptions, see Theorem 1 in \cite{LusPages06}).

%\newtheorem{thm}{Theorem}[section]

% page 631 libro Gilles sezione 7.2 condizioni esistenza soluzione debole crescita sublineare per a??
% l'idea e' di scrivere il teorema con le condizioni standard che i coefficienti devono soddisfare per garantire una soluzione debole, poi negli esempi dal momento che queste condizioni non sono soddisfatte bisognera' provare direttamente l'esistenza di una soluzione forte

\begin{prop}[Quantization of $Y$]
 $(a)$ {\em Functional codewords}. Assume that $h,g_1,g_2,b, a$ are Borel measurable functions satisfying, for all $(t,y,\tilde{y}^{g_1},\tilde{y}^{g_2},\tilde{y}^{h})\in [0,T]\times \mathbb{R}^4$, the linear growth conditions
\begin{align*}
\vert b(t,y, \tilde{y}^{g_1},\tilde{y}^{g_2})\vert &\leq C (1+\vert (y, \tilde{y}^{g_1},\tilde{y}^{g_2})\vert )\\
0<a(t,y, \tilde{y}^{h}) & \leq C (1+\vert (y, \tilde{y}^{h})\vert)\\
g_2(t,y) &\leq C (1+\vert y\vert ).
\end{align*}
Assume that $a$ is continuously differentiable with respect to both arguments and $ g_2\!\in C^1([0,T]\times\R)$. Then, the functional codewords of the quantizer of a weak solution  $Y=(Y_t)_{t\in [0,T]}$ of the SDE~\eqref{initialprocess} solve the following system of ODEs
\begin{align}\label{eq:codewordY}
\begin{split}
dy_t &=\left(b(t,y_t,\tilde{y}_t^{g_1},\tilde{y}_t^{g_2}) - \frac{1}{2}a(t,y_t, \tilde{y}_t^h)\frac{\partial a}{\partial y}(t,y, \tilde{y}^h_t)\right)dt + a(t,y_t, \tilde{y}_t^h)\alpha'(t)dt\\
d\tilde{y}_t^{g_2}&= g_2(t,y_t)\alpha'(t)dt  - \frac{1}{2}a(t,y_t, \tilde{y}_t^h)\frac{\partial g_2}{\partial y}(t,y_t)dt,
\end{split}
\end{align}
where $\alpha$ is a codeword of $W$.

\noindent $(b)$ {\em Functional quantization of $Y$}. Let  $ \Gamma= \{\alpha_1,\ldots,\alpha_{_N}\}$ be a quantizer of $W$ of size $N\in \mathbb{N}$ with  Voronoi diagram $(C_i(\Gamma))_{i=1,\ldots,N}$, then  the  quantization of  $Y$  is given by
\begin{equation*}
\widehat Y = \widehat Y^{W,\Gamma}:=\sum_{i=1}^N y^{(i)}\mbox{\bf 1}_{\{W\in C_i(\Gamma)\}}.
\end{equation*}
\end{prop}
%\textcolor{black}{[WHICH\ ASSUMPTIONS\ FOR\ $g_2$?? From \cite{PagSellami2011} page 285 it seems  that it should be continuously differentiable with bounded differentials, a priori no ellipticity required.. ]}

\begin{proof}
$(a)$ The functional codewords of the quantizer for $X$ in~\eqref{ODEX} are the solution of the following ODE:
\begin{align}
dx_t &= \left(\frac{b(t,S^{-1}_y(t,x,\tilde{y}_t^h),\tilde{y}_t^{g_1},\tilde{y}_t^{g_2})}{a(t,S^{-1}_y(t,x,\tilde{y}_t^h), \tilde{y}^h_t)}
+ \frac{\partial S}{\partial \tilde{y}^h}(t,S^{-1}_y(t,x,\tilde{y}_t^h),\tilde{y}^h_t)h(t,S^{-1}_y(t,x,\tilde{y}_t^h)) 
\right.\nonumber\\
&\quad \left. + \frac{\partial S}{\partial t}(t,S^{-1}_y(t,x,\tilde{y}_t^h),\tilde{y}^h_t) 
 - \frac{1}{2}\frac{\partial a}{\partial y}(t,S^{-1}_y(t,x,\tilde{y}_t^h), \tilde{y}^h_t)\right)dt  + \alpha'(t)dt,\label{ODEx}
\end{align}
where $\alpha(t)$ denotes a codeword of a quantizer $\Gamma$ at level $N$ for the Brownian motion, and  
$\tilde{y}_t^{g_2}$ denotes the functional codeword  for the process $ \tilde{Y}_t^{g_2}$ associated to $\alpha$, to be defined later on.
From $y_t = S^{-1}_y(t,x_t,\tilde{y}_t^h)$, we derive:
\begin{equation}\label{ODEy}
dy_t = a(t,y_t, \tilde{y}_t^h)dx_t + \frac{\partial S^{-1}_y}{\partial \tilde{y}^h}(t,x_t,\tilde{y}_t^h)h(t,y_t)dt+ \frac{\partial S^{-1}_y}{\partial t}(t,x_t,\tilde{y}_t^h)dt.
\end{equation}
The partial derivative $\frac{\partial S^{-1}_y}{\partial \tilde{y}^h}(t,x,\tilde{y}_t^h)$ (resp. $\frac{\partial S^{-1}_y}{\partial t}(t,x,\tilde{y}_t^h)$) can be deduced by differentiating the identity $S(t,S^{-1}_y(t,x,\tilde{y}^h),\tilde{y}^h) = x$ with respect to $\tilde{y}^h$ (resp. with respect to $t$), yielding:
\begin{equation}\label{identity}
\frac{\partial S^{-1}_y}{\partial \tilde{y}^h}(t,x,\tilde{y}^h) = -a(t,y,\tilde{y}^h)\frac{\partial S}{\partial \tilde{y}^h}(t,y,\tilde{y}^h), \quad \frac{\partial S^{-1}_y}{\partial t}(t,x,\tilde{y}^h) = -a(t,y,\tilde{y}^h)\frac{\partial S}{\partial t}(t,y,\tilde{y}^h)
\end{equation}
Substituting the ODE for $x_t$~\eqref{ODEx} into~\eqref{ODEy}, we obtain:
\begin{align*}
dy_t &= \left(b(t,y_t,\tilde{y}_t^{g_1},\tilde{y}_t^{g_2}) - \frac{1}{2}a(t,y_t, \tilde{y}_t^h)\frac{\partial a}{\partial y}(t,y, \tilde{y}^h_t) \right. \\
&\quad \left. + a(t,y_t, \tilde{y}_t^h)\frac{\partial S}{\partial \tilde{y}^h}(t,y_t,\tilde{y}^h_t)h(t,y_t)dt
+ a(t,y_t, \tilde{y}_t^h)\frac{\partial S}{\partial t}(t,y_t,\tilde{y}^h_t)\right)dt \\
&\quad + a(t,y_t, \tilde{y}_t^h)\alpha'(t)dt + \frac{\partial S^{-1}_y}{\partial \tilde{y}^h}(t,x,\tilde{y}^h)h(t,y_t)dt + \frac{\partial S^{-1}_y}{\partial t}(t,x,\tilde{y}^h)dt \\
&= \left(b(t,y_t,\tilde{y}_t^{g_1},\tilde{y}_t^{g_2}) - \frac{1}{2}a(t,y_t, \tilde{y}_t^h)\frac{\partial a}{\partial y}(t,y, \tilde{y}^h_t)\right)dt + a(t,y_t, \tilde{y}_t^h)\alpha'(t)dt,
\end{align*}
where the last equality follows from~\eqref{identity}.
 To complete the expression of the codeword $y$, we move to the functional codeword of $ \tilde{Y}_t^{g_2}$ associated to the code   $\alpha $ of $W$. The stochastic integral  $\tilde{Y}_t^{g_2}$ can be rewritten  in terms of the Stratonovich integral since $g(t,Y_t)$ is a semi-martingale as follows
\begin{align*}
\int_0^t g_2(u,Y_u)dW_u &=\int_0^t g_2(u,Y_u)\circ dW_u  - \frac{1}{2}\langle g_2(.,Y_.),W_.\rangle_t .
\end{align*}
It\^o's Lemma applied to   $g_2\!\in C^{1,2}([0,T]\times\R)$ yields
\begin{align*}
d g_2(t,Y_t) &=(\cdots)dt+ 
\textcolor{black}{a(t,Y_t, \tilde{Y}_t^h)}\frac{\partial g_2}{\partial y}(t,Y_t)dW_t,
\end{align*}
so that  $ 
\langle g_2(.,Y_.),W_.\rangle_t=\int_0^t a(u,Y_u, \tilde{Y}_u^h)\frac{\partial g_2}{\partial y}(u,Y_u)du$
and
\begin{align}
\int_0^t g_2(u,Y_u)dW_u =\int_0^t g_2(u,Y_u)\circ dW_u  - \frac{1}{2}\int_0^t a(u,Y_u, \tilde{Y}_u^h)\frac{\partial g_2}{\partial y}(u,Y_u)du.\label{sellami}
\end{align}
 We then apply \textcolor{black}{Theorem 2 in} \cite{PagSellami2011} in which is shown (see also~\cite{WongZakai1965}) that the appropriate way to quantize the above Stratonovich integral  is to associate  the  codeword $\alpha$ of $W$, namely the ``naive" codeword $\int_0^t g_2(u, y_u)\alpha'(u)du$.  This leads to define the functional codeword of the above  It\^o stochastic integrals  by
\begin{align*}
\tilde y^{g_2}_t = \int_0^t g_2(u, y_u)\alpha'(u)du- \frac{1}{2}\int_0^t a (u,y_u,\tilde{y}_u^{h})\frac{\partial g_2 }{\partial  y}(u,y_u)du,
\end{align*}
where $\alpha$ is a codeword of a quantizer of $W$. The expression~\eqref{eq:codewordY}  of the functional codeword of $Y$ attached to the functional codeword $\alpha$ of $W$ immediately follows. 

\noindent $(b)$ \textcolor{black}{Using rough path theory,  in a larger space of paths which   involves the so-called ``enhanced path of the Brownian motion" (also called geometric multiplicative functional lying on $W$), \cite{PagSellami2011} showed that the quantizations of $Y$ written on the Voronoi diagrams of a quantization $\widehat  W^N$ of $W$ will converge to  $Y$ when $W^N$ converges to $W$ as $N\to+\infty$  for some H\"older norms. Such is the case if $\widehat W^N$ is a sequence of quadratic optimal product quantizers as defined in Section~\ref{subsec:FPQW}, see~\cite[Section~3.1]{PagSellami2011}. Then,  the result follows from the convergence to   $Y$ of its  quantizations  written on the Voronoi diagrams of  $\widehat  W^N$  as $N\to+\infty$.}
\end{proof}
Note that the expression~\eqref{eq:codewordY} is quite similar to that obtained with regular diffusion, including the observed connection with stochastic integration in the Stratonovich sense.
In conclusion, we have therefore arrived at ODEs, or rather a "bundle" of ODEs, where the $\alpha_i$ are the functional codewords of an optimal quantization of $W$. These ODEs can be solved using a numerical scheme, such as e.g. a Runge-Kutta type method or other schemes, like e.g. in  \cite{pagesprintem05}, where they  point
out the efficiency of the Romberg log-extrapolation (sometimes combined with a linear interpolation
method) which numerically outperforms Monte Carlo simulation for moderate values of $N$, say less than $10\,000$.

Let us emphasize again that the Lamperti transformation is a demonstrative tool that allows  to obtain a priori error estimates of the type $O((\log n)^{-1/2})$ for regular (Markovian) diffusions with Lipschitz coefficients. It also requires a   uniform ellipticity assumption on the diffusion coefficient. This is illustrated by  the proofs of \cite[Section 7.2, p.610]{LusPag2023}, see also~\cite{LusPages06}. In the present much more general approach, obtaining similar error bounds is more demanding essentially due to the presence of the stochastic integral $\tilde Y^{g_2}_t$ in the drift coefficient $b$. Approximating such a stochastic integral needs to switch to a Stratonovich formulation  to call upon rough path theory as mentioned above. Doing so, we follow~\cite{PagSellami2011} , still in the spirit of~\cite{WongZakai1965}. Unfortunately, the resulting approximation  by functional quantization, though converging,  did not give rise to error estimates so far, to our best knowledge. This question is strongly connected to the Lipschitz regularity of the so-called It\^o map (see~\cite{Lejay2009} for a survey).
Also, it is important to note that the  ODEs in the proposition refer to functional  codewords of a  quantizer of $Y$ where the cells of the tessellation and the weights are suboptimal  (as they represent a  Voronoi diagram for $W$ only). Actually, calculating the optimal cells and weights for $Y$ would require a massive MC simulation, involving Lloyd's algorithm. We keep such a goal for future research which goes beyond the scope of the present paper. However, in the case  -- not satisfied by the models \textcolor{black}{investigated} in the  present paper ---  where the coefficients of the SDE do not depend on a  stochastic integral  term like $\tilde Y^{g_2}$  (but possibly on standard temporal means of $Y$ like $\tilde Y^{g_1}$), then under some regularity assumptions (e.g., Lipschitz assumptions on the coefficients wrt their arguments uniformly in $t\!\in[0,T]$ and uniform ellipticity  of the diffusion coefficient $a$), a straightforward adaptation of the proof of \cite[Theorem~7.2]{LusPag2023}] yields a quantization rate of $O\big((\log n)^{-1/2}\big)$ when using optimal product quantizations of $W$.

In the next section, we will apply the methodology to the volatility models described in Section \ref{section2}. It is imperative to verify, on a case-by-case basis, the assumptions that ensure the existence of at least a weak solution for the associated SDEs, thereby warranting the application of the entire procedure. Remarkably, thanks to a powerful result on SDEs with path-dependent coefficients (see \cite{rogerswill2000}), we will be able to prove the existence and uniqueness of a strong solution to the SDEs of all the models considered in this paper. These results represent an independent contribution for our paper. In particular, in the case of the model of \cite{GuyonVolMostlyPathDep2022}, our proof is  alternative to the one provided by  \cite{NutzValde2023} and it is also simpler (as could be expected) in the case of an exponential kernel than that  by \cite{AndresJourdain2024} for more general kernels.

\textcolor{black}{We end this section by mentioning the  computational performance of functional quantization, that, as expected,  proves to be significantly more efficient than Monte Carlo simulation. We note that  in the present paper we did not consider acceleration techniques such as the Romberg log-extrapolation employed in \cite{pagesprintem05} in the context of pricing path-dependent options. This choice is motivated partly by the fact that such numerical improvements go beyond the illustrative scope of our numerical experiments, and, more importantly, because the present work focuses on particularly simple European-type products for which pricing via quantization is already extremely efficient. As a reference, the most computationally demanding routine corresponds to the solution - via a fourth-order Runge-Kutta scheme - of the system of ODEs satisfied by the codewords, which is completed in less than one second for $N\approx 158,000$.}

\section{Application to the Model  of \cite{GuyonVolMostlyPathDep2022}}\label{section4}

\subsection{Existence and Uniqueness of a Strong Solution}

In this subsection, we apply Theorem 11.2 from \cite{rogerswill2000} (p. 128), which proves to be instrumental in establishing the existence and uniqueness of the solution to the system of stochastic differential equations (SDEs) defined by~\eqref{R1dyn} and~\eqref{R2dyn}, when the volatility process is specified by~\eqref{volJulien}.
We will apply this theorem to our system of SDEs to establish the existence and uniqueness of a strong solution for the path-dependent volatility model of \cite{GuyonVolMostlyPathDep2022}.

From~\eqref{R2dyn} we get
\begin{align*}
d\sqrt{R_{2,t}}&= \frac{1}{2\sqrt{R_{2,t}}}dR_{2,t}= \frac{\lambda_2}{2}\frac{\sigma_t^2-R_{2,t}}{\sqrt{R_{2,t}}}dt,
\end{align*}
from which we have 
\begin{align*}
\sqrt{R_{2,t}}&= \sqrt{R_{2,0}} +  \frac{\lambda_2}{2}\int_0^t
\Phi (u,\sigma_.)du,
\end{align*}
\textcolor{black}{(with  the notational convention from~\cite{rogerswill2000}  $\Phi (u,\sigma_.)= \Phi(u,(\sigma_v)_{0\le v\le u})$ that we will adopt throughout this section for functionals of this form)} where 
\begin{align*}
\Phi (u,\sigma_.)&=\frac{\sigma_u^2-R_{2,0}e^{-\lambda_2u}-\lambda_2\int_0^ue^{-\lambda_2(u-s)}\sigma_s^2ds}{\left(R_{2,0}e^{-\lambda_2u}+\lambda_2\int_0^ue^{-\lambda_2(u-s)}\sigma_s^2ds\right)^{1/2}},
\end{align*}
obtained by inserting  the solution~\eqref{R2t} of~\eqref{R2dyn}.

Now, from~\eqref{volJulien} we get
\begin{align}\label{eq:EDSsigma}
\sigma_t &= \beta_0+\beta_1R_{1,0}e^{-\lambda_1t}
+\beta_1\lambda_1\int_0^t e^{-\lambda_1(t-u)}\sigma_u dW_u
+\beta_2\sqrt{R_{2,0}}
+\frac{\beta_2\lambda_2}{2}\int_0^t
\Phi (u,\sigma_.)du. 
\end{align}
%where we used that $e^{-\lambda_1t}=1-\int_0^t \lambda_1e^{-\lambda_1u}du$. 
%
Let us introduce $\tilde\sigma_t:=\sigma_t e^{\lambda_1 t}$, so that
\begin{align}\label{eq:EDStildesigma}
\tilde\sigma_t &= \sigma_0e^{\lambda_1t}
+\beta_1\lambda_1\int_0^t e^{\lambda_1u}\sigma_u dW_u
-\beta_1R_{1,0}(e^{\lambda_1t}-1)
+\frac{\beta_2\lambda_2}{2}e^{\lambda_1t}\int_0^t
\Phi (u,\sigma_.)du,
\end{align}
where $\sigma_0=\beta_0+\beta_1R_{1,0}+\beta_2\sqrt{R_{2,0}}$.
% and we replaced $\lambda_1\int_0^t e^{\lambda_1(t-u)}du=e^{\lambda_1t}-1$.
%
Note that 
\begin{align*}
\Phi (u,\sigma_.)&=\frac{\tilde\sigma_u^2e^{-2\lambda_1 u}-R_{2,0}e^{-\lambda_2u}-\lambda_2e^{-\lambda_2 u}\int_0^ue^{(\lambda_2 -2\lambda_1)s}\tilde\sigma_s^2ds}{\left(R_{2,0}e^{-\lambda_2u}+\lambda_2e^{-\lambda_2 u}\int_0^u e^{(\lambda_2 -2\lambda_1)s}\tilde\sigma_s^2ds\right)^{1/2}}=:\tilde\Phi (u,\tilde\sigma_.).
\end{align*}
Therefore \textcolor{black}{$\tilde\sigma_0=\sigma_0$ and}
\begin{align*}
d\tilde\sigma_t &= \left(\lambda_1 \sigma_0e^{\lambda_1t}
-\beta_1\lambda_1R_{1,0} e^{\lambda_1t}
+\frac{\beta_2\lambda_1\lambda_2}{2}e^{\lambda_1t}\int_0^t
\tilde\Phi (u,\tilde\sigma_.)du
+\frac{\beta_2\lambda_2}{2}e^{\lambda_1t}\tilde\Phi (t,\tilde\sigma_.)\right) dt
+\beta_1\lambda_1\tilde\sigma_t dW_t\\
 &= \lambda_1\left( \sigma_0-\beta_1R_{1,0}\right) e^{\lambda_1t}dt
+\frac{\beta_2\lambda_2}{2}e^{\lambda_1t}\left(\lambda_1\int_0^t
\tilde\Phi (u,\tilde\sigma_.)du+\tilde\Phi (t,\tilde\sigma_.)\right) dt
+\beta_1\lambda_1\tilde\sigma_t dW_t.
\end{align*}

The previous equation appears  now as a path-dependent  SDE in  the form required by \cite[Theorem 11.2]{rogerswill2000},  that is
\begin{align}
\tilde\sigma_t&=\sigma_0 + \int_0^t \tilde b(u,\textcolor{black}{\tilde\sigma_.})du +\int_0^t \tilde\nu (u,\textcolor{black}{\tilde\sigma_.})dW_u,\label{sigmaRW}
\end{align}
where, for all  $x\in {\cal C}\left([0,T],\mathbb{R}\right)$,
\begin{align*}
\tilde b(t,x)&= 
 \lambda_1\hskip-0.1cm\left( \sigma_0-\beta_1R_{1,0}\right) e^{\lambda_1t}
\!+\frac{\beta_2\lambda_2}{2}e^{\lambda_1t}\Big(\lambda_1\hskip-0.15cm \int_0^t
\tilde\Phi (u,x_.)du+\tilde\Phi (t,x_.)\hskip-0.1cm\Big) \mbox{ and } \tilde\nu (t,x)=\beta_1\lambda_1 x(t),
\end{align*}
and $\sigma_0=\tilde\sigma_0=\beta_0+\beta_1R_{1,0}+\beta_2\sqrt{R_{2,0}}$.
\textcolor{black}{Note that $\tilde b(t,x)$ and $\tilde\nu(t,x)$ only depend on $x(\cdot \wedge t)$ and are consequently {\em predictable  path-dependent functionals} in the sense of Rogers and Williams' Definition~(6.3). Also note that  $\nu (t,0)=0$ and $\tilde\Phi (t,0)=-\sqrt{R_{2,0}}e^{-\frac{\lambda_2}{2}t}$, so that Condition~(11.3) (that is $\vert \nu (t,0)\vert +\vert b(t,0)\vert \leq C_T$, $t\!\in [0,T]$) of Rogers and Williams' Theorem 11.2  is satisfied. 
}
We will first apply this theorem to a  version of~\eqref{sigmaRW} where we bound the coefficients using a classic localization argument. For this purpose,  let us introduce, for every integer $K\ge 2$, the ${\cal C}^\infty$ function $\phi_K : \mathbb{R}\rightarrow  [-K,K]$ such that $\| \phi_K'\|_{\sup}\leq 1$ and 
\begin{align*}
\phi_K (x)=x \quad \hbox{for}\ \vert x\vert \leq K-1;\quad 
\vert \phi_K (x)\vert =K \ \ \ \hbox{for}\ \vert x\vert \geq K.
\end{align*}

We can then show that, for all  $x\!\in {\cal C}\left([0,T],\mathbb{R}\right)$, the coefficients 
\begin{align*}
\tilde b^{(K)}(t,x)&= 
 \lambda_1\left( \sigma_0-\beta_1R_{1,0}\right) e^{\lambda_1t}
+\frac{\beta_2\lambda_2}{2}e^{\lambda_1t}\left(\lambda_1\int_0^t
\tilde\Phi (u,\phi_K(x_.))du+\tilde\Phi (t,\phi_K(x_.))\right), \\
\tilde \nu^{(K)} (t,x)&=\tilde \nu(t,x) = \beta_1\lambda_1 x(t)
\end{align*}
are uniformly Lipschitz continuous
% with uniformly in $t\in[0,T]$ 
in the  following sense: for all $K\ge 2$, $t\!\in [0,T]$, and $x$, $y\!\in  {\cal C}\left([0,T],\mathbb{R}\right)$,
 \begin{equation}\label{eq:UniLip}
|\tilde b^{(K)}(t,x)-\tilde b^{(K)}(t,y) |+ |\tilde \nu^{(K)}(t,x)-\tilde \nu^{(K)}(t,y)|\le C_{b,\nu,T}\|x-y\|_{[0,t]},
\end{equation}
where $\|h\|_{[0,t]} := \sup_{0\le s\le t}|h(s)]$. This is obvious for $\tilde \nu^{(K)} (t,x)$, while for $\tilde b^{(K)}(t,x)$ 
 we observe that the map $\tilde b^{(K)}_{t}:x\rightarrow \tilde b^{(K)}(t,x)$ is differentiable on ${\cal C}\left([0,T],\mathbb{R}\right)$ and for any $h\in {\cal C}\left([0,T],\mathbb{R}\right)$ we obtain
\begin{align*}
Db^{(K)}_{t}(x)\cdot h&=
\frac{\beta_2\lambda_2}{2}e^{\lambda_1t}\left(\lambda_1\int_0^t
D\tilde\Phi (u,\phi_K(x_.))\cdot h(u)du+D\tilde\Phi (t,\phi_K(x_.))\cdot h(t)\right).
\end{align*}
Now
{\footnotesize
\begin{align*}
&D\tilde\Phi (u,\phi_K(x_.))\cdot h\\
&=
\frac{-2 \lambda_2e^{-\lambda_2 u}\int_0^u e^{(\lambda_2-2\lambda_1)s}\phi_K(x(s))\phi_K'(x(s))h(s)ds}{\sqrt{R_{2,0}e^{-\lambda_2 u}
+ \lambda_2e^{-\lambda_2 u}\int_0^u e^{(\lambda_2-2\lambda_1)s}\phi_K(x(s))^2ds}}\\
&-\frac{\left(\phi^2_K(x(t))e^{-2\lambda_1 u} \hskip-0.15cm -R_{2,0}e^{-\lambda_2 u}\hskip-0.15cm- \lambda_2e^{-\lambda_2 u} \hskip-0.15cm\int_0^u e^{(\lambda_2-2\lambda_1)s}\phi^2_K(x(s))ds
\right) \hskip-0.1cm\lambda_2e^{-\lambda_2 u}\int_0^u e^{(\lambda_2-2\lambda_1)s}\hskip-0.1cm\phi_K(x(s))\phi_K'(x(s))h(s)ds}
{\left(R_{2,0}e^{-\lambda_2 u}
+ \lambda_2e^{-\lambda_2 u}\int_0^u e^{(\lambda_2-2\lambda_1)s}\phi_K(x(s))^2ds\right)^{3/2}},
\end{align*}
}
so that
%if we set $\|h\|_{[0,u]}=\sup_{s\in [0,u]}|h(s)|$, then 
\footnotesize
\begin{align*}
\| D\tilde\Phi (u,\phi_K(x_.))\cdot h\|_{\sup}&\leq 
 2\lambda_2e^{-\lambda_2 u}
 \frac{
 \left(\int_0^u e^{2(\lambda_2-2\lambda_1)s}\phi_K^2(x(s))(\phi_K'(x(s)))^2ds\right)^{1/2}
\left(\int_0^uh^2(s)ds\right)^{1/2}}
{\left(R_{2,0}e^{-\lambda_2 u}\right)^{1/2}}\\
 &\quad +\lambda_2e^{-\lambda_2 u}K^2
\frac{
 \left(\int_0^u e^{2(\lambda_2-2\lambda_1)s}\phi_K^2(x(s))(\phi_K'(x(s)))^2ds\right)^{1/2}
\left(\int_0^uh^2(s)ds\right)^{1/2}
}{\left(R_{2,0}e^{-\lambda_2 u}\right)^{3/2}}\\
&\leq \| h\|_{\sup}\frac{\lambda_2 \textcolor{black}{u^{\frac{1}{2}}}e^{-\lambda_2 u}K^2
\left(\int_0^u e^{2(\lambda_2-2\lambda_1)s}ds\right)^{1/2}}{\Big(R_{2,0}e^{-\lambda_2 u}\Big)^{1/2}}
\Big(2+\frac{K^2}{R_{2,0}e^{-\lambda_2 u}}\Big) =\|h\|_{[0,u]}f(u,K).
\end{align*}

\normalsize
In conclusion, we clearly have
\begin{align*}
\vert Db^{(K)}_{t}(x)\cdot h\vert &\leq\|h\|_{[0,t]}
\frac{\beta_2\lambda_2}{2}e^{\lambda_1t}\left(\lambda_1\int_0^t
f(u,K)du+f(t,K)\right),
\end{align*}
\textcolor{black}{therefore also the coefficient $\tilde b^{(K)}$ and  $\nu^{(K)}$ are   uniformly Lipschitz   in the sense of~\eqref{eq:UniLip}, so that Rogers and Williams'  Theorem 11.2  applies to the process, showing strong existence and uniqueness of $\tilde \sigma^{(K)}_t$ solution to the path-dependent SDE 
\begin{align*}
\tilde\sigma^{(K)}_t&=\sigma_0 + \int_0^t \tilde b^{(K)}(u,\tilde\sigma^{(K)}_{.\wedge t})du +\int_0^t \tilde \nu^{(K)} (u,\tilde\sigma^{(K)}_{.\wedge t})dW_u, \quad t\ge 0.
\end{align*}
%which in turns implies existence and uniqueness of the solution to the equation 
%\begin{align*}
%\sigma^{(K)}_t&=\sigma_0 + \int_0^t b^{(K)}(u,\sigma^{(K)}_{.\wedge t})du +\int_0^t \nu^{(K)} (u,\sigma^{(K)}_{.\wedge t})dW_u
%\end{align*}
Now we set  $\tilde \tau_{K}:=\inf \{t: \vert \tilde \sigma_t^{(K)}\vert \geq K-1\}$. One easily checks  by the definition of functionals $\Phi_K$ and the local property of stochastic integral w.r.t. stopping times and the fact that $\tilde \tau_k\le \tilde \tau_{k+2}$ that
\[
\mathbb{P}\mbox{-}a.s.\quad \tilde  \sigma^{(K+1)}_{\cdot\, \wedge \tilde \tau_{K}} =\tilde \sigma^{(K)}_{\cdot\, \wedge \tilde \tau_{K}}.
\]
Taking advantage of this telescopic feature we may define in a consistent way 
\begin{equation}\label{eq:tilde-sigma}
\tilde \sigma_t = \tilde \sigma^{(K)}_t \quad \mbox{ on }\quad \{t\le \tau_K\}.
\end{equation}
From~\eqref{volJulien},~\eqref{R1t} and~\eqref{R2t}, we get for every $K\ge 2$ and $t\ge 0$,
\begin{align*}
\tilde \sigma^{(K)}_{t\wedge \tilde \tau_K} &= \beta_0e^{\lambda_1(t\wedge \tilde \tau_K)}+\beta_1\Big(R_{1,0}+\lambda_1\int_0^{t\wedge \tilde \tau_K}\tilde \sigma_udW_u\\
&\qquad +\beta_2e^{(\lambda_1-\lambda_2/2)(t\wedge \tilde \tau_K)} \big(R_{2,0} + \int_0^{t\wedge \tilde \tau_K} e^{-(2\lambda_1-\lambda_2)u} \tilde \sigma_u^2du\big)^{1/2}\Big).
\end{align*}
Hence, using the $\tilde \sigma_u = \sigma^{(K)}_u$ when $u\le s\wedge \tilde \tau_K$,
%\small
\begin{align*}
g_K(t) := \mathbb{E}\left[\sup_{s\le t} (\tilde\sigma^{(K)}_t)^2\right]&\leq  3^2\Big( (\beta_0e^{\lambda_1 t}\hskip-0.5cm  +\beta_1R_{1,0})^2+\beta_2^2 e^{|2\lambda_1-\lambda_2|t}R_{2,0}+       \beta_1^2\lambda^2_1\mathbb{E}\left[\sup_{s\le t}\Big(\int_0^{s\wedge \tilde \tau_K} \hskip-0.25cm \tilde \sigma_udW_u\Big)^2\right]\\
 &\qquad + \beta_2^2 e^{|2\lambda_1-\lambda_2|t}  \lambda_2\E\left[\,\int_0^{t\wedge \tilde \tau_K} \hskip-0.25cm(\tilde\sigma^{(K)}_u)^2du\right]\Big).
\end{align*}
\normalsize
Then applying Doob's inequality   to the  stopped stochastic integral term, Fubini's Theorem  and the fact that $\mathbb{E}\left[\big(\int_0^{t\wedge \tilde \tau_k} \tilde \sigma^2_u du\big)^2\right] \le\int_0^t \mathbb{E}\left[\tilde \sigma^2_{u\wedge  \tilde \tau_K}\right]du\le  \int_0^t g_K(u)du$, finally yields 
$g_K(t)\le C_1(t) +C_2(t)\int_0^t g_K(u)du$
 with $C_1(t)=9\Big( (\beta_0e^{\lambda_1 t}\! +\beta_1R_{1,0})^2\!+\beta_2^2 e^{|2\lambda_1-\lambda_2|t}R_{2,0}\Big)$ and $C_2(t) = 9(4  \beta_1^2\lambda^2_1+ \beta_2^2 e^{|2\lambda_1-\lambda_2|t})$ which are both  increasing in $t$ and do not depend on $K$. 
Gr\"onwall's Lemma then  yields 
$$ 
\mathbb{E} \left[\sup_{s\le t\wedge \tilde \tau_K}(\tilde\sigma_s)^2 \right]= \mathbb{E}\left[\sup_{s\le t} (\tilde\sigma^{(K)}_{s\wedge \tilde \tau_K})^2\right]= g_K(t)\le C_1(t)e^{C_2(t)t},\; t\ge 0.
$$
As a consequence, for every integer $K\ge 2$, 
\[
\mathbb{P}(\tilde \tau_K \le t) \le  \mathbb{P}\big(\sup_{0\le s\le t\wedge \tilde \tau_K }|\tilde \sigma^{(K)}_s|\ge K-1\big) \le \frac{g_K(t)}{(K-1)^2}\le \frac{C_1(t) e^{ C_2(t)t}}{(K-1)^2}
\]
which in turn implies that $\tilde \tau_\infty= \lim_{K} \tilde \tau_K= \sup_{K\ge 2} \tilde \tau_K$ satisfies 
\[
\forall\, t\ge 0, \quad \mathbb{P}(\tilde \tau_\infty \le  t) = \lim_{K}  \mathbb{P}(\tilde \tau_K\le t)= 0 \quad \mbox{ i.e. } \quad \tilde \tau_\infty = +\infty\; \mathbb{P}\mbox{-}a.s. 
\]
Hence  $\tilde \sigma_t$   defined by~\eqref{eq:tilde-sigma} is well-defined over $[0, +\infty)$ (up to a negligible set) since $\tilde \tau_\infty$ is clearly its explosion time and is solution to~\eqref{eq:EDStildesigma}. Uniqueness is obvious  by localization. Then we set $\sigma_t = e^{\lambda_1t}\tilde \sigma_t$  and we straightforwardly check  using It\^o's Lemma that it is solution to~\eqref{eq:EDSsigma}.\hfill$\Box$
}

\medskip

\begin{Remark}
\textcolor{black}{The proof of strong existence and uniqueness in the framework of Volterra equations with power kernels by  \cite{AndresJourdain2024} is based on different localization arguments (the counterpart of Theorem 11.2 in \cite{rogerswill2000} does not yet exist for such SDE). 
 For example,   \cite{AndresJourdain2024} show that, from a practical point of view,  it is important to use two time scales in $R_{2,t}$, while a single exponential kernel is sufficient in $R_{1,t}$ to ensure a good calibration of the \cite{GuyonVolMostlyPathDep2022} model. Such extensions can be easily taken into account in our  proof for existence and uniqueness.
Finally, our approach enables the development of efficient numerical methods, which is considerably more challenging in the framework of \cite{AndresJourdain2024} that cannot be reduced to finite-dimensional Markovian models.}
\end{Remark}

\subsection{Parametric Restrictions for the Positivity of the Volatility Process}

The positivity of the volatility process in the model of \cite{GuyonVolMostlyPathDep2022} requires some parametric restrictions. In fact, following this time  \cite{NutzValde2023}, we can provide a brief  sketch of the proof by taking
\begin{align*}
d\ln \sigma_t&=\beta_1\lambda_1dW_t -\frac{1}{2}\beta_1^2\lambda_1^2 dt+
\frac{1}{\sigma_t}\left( -\beta_1 \lambda_1 R_{1,t}+
\frac{\beta_2\lambda_2}{2}\frac{\sigma_t^2-R_{2,t}}{\sqrt{R_{2,t}}}\right)dt
\nonumber\\
&=\beta_1\lambda_1dW_t -\frac{1}{2}\beta_1^2\lambda_1^2 dt
%\\&\quad 
+
\frac{1}{\sigma_t}\left( - \lambda_1 \sigma_t+ \beta_0\lambda_1
+\beta_2\left(\lambda_1-\frac{\lambda_2}{2}\right)\sqrt{R_{2,t}}
+\frac{\beta_2\lambda_2}{2}\frac{\sigma_t^2}{\sqrt{R_{2,t}}}\right)dt
\nonumber\\
&\geq\beta_1\lambda_1dW_t -\frac{1}{2}\beta_1^2\lambda_1^2 dt  - \lambda_1 dt
\end{align*}
provided that $\lambda_2<2\lambda_1$.
It follows that, if $\sigma_0>0$ and $\lambda_2<2\lambda_1$, the volatility remains bounded away from zero, in fact
$ \sigma_t\geq\sigma_0 e^{\beta_1\lambda_1W_t -\frac{1}{2}\beta_1^2\lambda_1^2 t}e^{-\lambda_1 t}>0$.
This can be made rigorous by appropriate localization arguments.

\subsection{Lamperti Transform}

From~\eqref{aJulien} with $\epsilon>0$, namely 
\begin{align*} a\Big(t,y_t,\int_{\epsilon}^t h(u,y_u)du\Big)=\beta_1\lambda_1 y_t 
\end{align*}
we get that there is no dependence on $\tilde y_t^h$ in the function $S$:
\begin{equation*}
S(y,\tilde y^h)= \int_{\epsilon}^y \frac{d\xi}{a (\xi, \tilde y^h)}=\int_{\epsilon}^y \frac{d\xi}{\beta_1\lambda_1\xi}=
\frac{1}{\beta_1\lambda_1}\ln \left( \frac{y}{\epsilon}\right),
\end{equation*}
so that the inverse is simply given by
\begin{equation*}
y=S^{-1}_y(x,\tilde y^h)= \epsilon \exp (\beta_1\lambda_1 x). 
\end{equation*}

Let us now consider the drift $b$ in~\eqref{bJulien}:
\begin{align*} 
b\Big(t,y_t,\tilde y_t^{g_1},\tilde y_t^{g_2}\Big)&= 
-\beta_1 \lambda_1 e^{-\lambda_1 t}\left( R_{1,0}+\lambda_1\tilde y_t^{g_2}\right)+
\frac{\beta_2\lambda_2}{2}\frac{y_t^2-e^{-\lambda_2 t}\left( R_{2,0}+ \lambda_2\tilde y_t^{g_1}\right)}{\sqrt{e^{-\lambda_2 t}\left( R_{2,0}+ \lambda_2\tilde y_t^{g_1}\right)}}
\end{align*}
so that the ODE satisfied by the codeword process $y_t$ becomes 
{\small  \begin{align}
d y_t&= \left(b(t,y_t,\tilde y_t^{g_1},\tilde y_t^{g_2})-\frac{1}{2}a (y_t, \tilde y_t^h)\frac{\partial a}{\partial y}(y, \tilde y^h_t)\right)dt
+a (y_t, \tilde y_t^h)\alpha'(t)dt\nonumber\\
&= \left(
-\beta_1 \lambda_1 e^{-\lambda_1 t}\left( R_{1,0}+\lambda_1\tilde y_t^{g_2}\right)+
\frac{\beta_2\lambda_2}{2}\frac{y_t^2-e^{-\lambda_2 t}\left( R_{2,0}+ \lambda_2\tilde y_t^{g_1}\right)}{\sqrt{e^{-\lambda_2 t}\left( R_{2,0}+ \lambda_2\tilde y_t^{g_1}\right)}}
-\frac{\beta^2_1\lambda^2_1}{2}y_t\right)dt
+\beta_1\lambda_1 y_t\alpha'(t)dt\label{codewardsJulien}
\end{align}
}
% while the ODE satisfied by the process $x_t$ becomes
 %\begin{align*}
%d x_t&=\left(
%\frac{b(t,S^{-1}_y(x_t,\tilde y_t^h),\tilde y_t^{g_1},\tilde y_t^{g_2})}{a (S^{-1}_y, \tilde y^h_t)}
%+\frac{\partial S}{\partial \tilde y^h}(S^{-1}_y(x_t,\tilde y_t^h),\tilde y^h_t)h(t,S^{-1}_y)
%-\frac{1}{2} \frac{\partial a}{\partial y}(S^{-1}_y(x_t,\tilde y_t^h), \tilde y^h_t)\right)dt
%+\alpha'(t)dt\\
%&= \left(- \lambda_1 e^{-\lambda_1 t} \frac{ \tilde y_t^{g_1}}{y_t}
%+\frac{\beta_2 \lambda_2}{\beta_1 \lambda_1}\frac{y_t-e^{-\lambda_2 t}\frac{\tilde y_t^{g_2}}{y_t}}{\sqrt{e^{-\lambda_2 t}\tilde y_t^{g_2}}}
%-\frac{\beta_1 \lambda_1}{2})\right)dt
%+\alpha'(t)dt\\
%&= \left(- \lambda_1 e^{-\lambda_1 t} \frac{ \tilde y_t^{g_1}}{\epsilon \exp (\beta_1\lambda_1 x_t)}
%+\frac{\beta_2 \lambda_2}{\beta_1 \lambda_1}\frac{\epsilon \exp (\beta_1\lambda_1 x_t)-e^{-\lambda_2 t}\frac{\tilde y_t^{g_2}}{\epsilon \exp (\beta_1\lambda_1 x_t)}}{\sqrt{e^{-\lambda_2 t}\tilde y_t^{g_2}}}
%-\frac{\beta_1 \lambda_1}{2})\right)dt
%+\alpha'(t)dt,
%\end{align*}
where we have 
\begin{align*}
\tilde y_s^{g_1}\!&\!=\int_{0}^s g_1(u,y_u)du= \int_{0}^t e^{\lambda_2 u}  y_u^2du,\\
\tilde y^{g_2}_t \!&\!= \int_0^t g_2(u, y_u)\alpha'(u)du- \tfrac{1}{2}\int_0^t\!\! a (\textcolor{black}{y_u},\tilde{y}_u^{h})\frac{\partial g_2 }{\partial  y}(u,y_u)du
%\\&
= \!\int_{0}^t e^{\lambda_1 u} y_u\alpha'(u)du-\frac{\beta_1\lambda_1}{2} \!\!\int_0^t e^{\lambda_1 u} y_udu.
\end{align*}
For the model proposed by \cite{GuyonVolMostlyPathDep2022}, the functional quantization method can be applied with remarkable simplicity due to the absence of an integral term in the diffusion coefficient.
This simplification in the application of functional quantization stands in contrast to more complex models where the diffusion term includes integral components, like the one we shall consider in the next section.

\subsection{Numerical Illustration for the Model of \cite{GuyonVolMostlyPathDep2022}}

We use the same parameters as in \cite{GuyonVolMostlyPathDep2022} for their 2-factor specification: $\beta_0=0.08;\ \beta_1=-0.08;\ \beta_2=0.5;\ \lambda_1=62;\ \lambda_2=40$ and $R_{1,0}=-0.044;\ R_{2,0}=0.007;\ S_0=100;\ r=0;\ T=1$ with $500$ steps for the time discretization and $
\sigma_0=y_0 = \beta_0 + \beta_1 R_{1,0} + \beta_2 \sqrt{R_{2,0}}=0.1254$.

For the spatial discretization, we retain $d_N=20$  coefficients in the Karhunen-Lo\`eve expansion, which are partitioned into five independent blocks,
$\{ (1:4); (5:8); (9:12); (13:16); (17:20)\}$, each consisting of independent four-dimensional Gaussian vectors. Vector quantization is then applied separately to each block, using $N_1=60$ points for the first block, $N_2=45$ for the second, and $N_3=9$,  $N_4=3$, $N_5=2$ for the remaining blocks, respectively. This allocation is clearly not unique and reflects a deliberate trade-off between approximation accuracy and computational complexity.

From an intuitive standpoint, the infinite-dimensional Brownian motion on $[0,1]$ is approximated by a finite-dimensional truncated Karhunen-Lo\`eve expansion with $d_N=20$ modes, which captures more than $99.5\%$ of the total variance of the process. Consequently, each Brownian path can be identified with a point in $\mathbb{R}^{20}$. Product functional quantization is then carried out on this finite-dimensional representation, yielding a total of $N=N_1\times N_2\times N_3\times N_4\times N_5 =157,950$ distinct quantized trajectories.

The spot volatility \eqref{volJulien}, that we recall for convenience:
\begin{equation*}
\sigma_t = \sigma(R_{1,t}, R_{2,t}) = \beta_0 + \beta_1 R_{1,t} + \beta_2 \sqrt{R_{2,t}}
\end{equation*}
is quantized through the codewords $y_t$ satisfying the  ODEs \eqref{codewardsJulien}. 
Now let us consider the functional quantization of the underlying process $S$ given by
\begin{align*}
S_t  &= S_0 \exp\left( -\frac{1}{2}\int_0^t Y_s^2ds+   \int_0^t  Y_s dW_s  \right)\\
&= S_0 \exp\left( -\frac{1}{2}\int_0^t Y_s^2ds+ \int_0^t Y_s \circ dW_s
 - \frac{1}{2}\int_0^t a(s,Y_s)\frac{\partial Y_s}{\partial y}ds \right)\end{align*}
where we applied \eqref{sellami} to
the Stratonovich  integral, so that the functional codeword for the underlying process is given by
\begin{align*}
\hat S_t  &= S_0 \exp\left( -\frac{1}{2}\int_0^t y_s^2ds+ \int_0^t y_s \alpha'(s)ds
 - \frac{\beta_1\lambda_1}{2}\int_0^t y_sds \right), 
\end{align*}
where we used that $a(s,Y_s)=\beta_1\lambda_1 Y_s$. We now have all the ingredients to solve the system of ODEs using a fourth-order Runge-Kutta scheme and employ the functional quantization framework introduced above to determine the Brownian codeword $\alpha$.

As a first step, we visualize some quantized trajectories for  $\hat S$  and $y$ associated to  the volatility process. In Figure~\ref{fig:Julien100pathSandY}, we also present $100$ trajectories of the original stochastic processes $S$ and $Y$ generated by Monte Carlo simulation using a standard Euler scheme of the corresponding SDEs, which serve both as a reference benchmark for all subsequent numerical results and as a consistency check of the methodology.

% 100 TRAIETTORIE S e y
%\usepackage{graphicx}
%\usepackage{subcaption}

\begin{figure}[!htbp]
\centering

% -------- Prima riga --------
\begin{minipage}{0.45\linewidth}
\centering
\includegraphics[width=\linewidth]{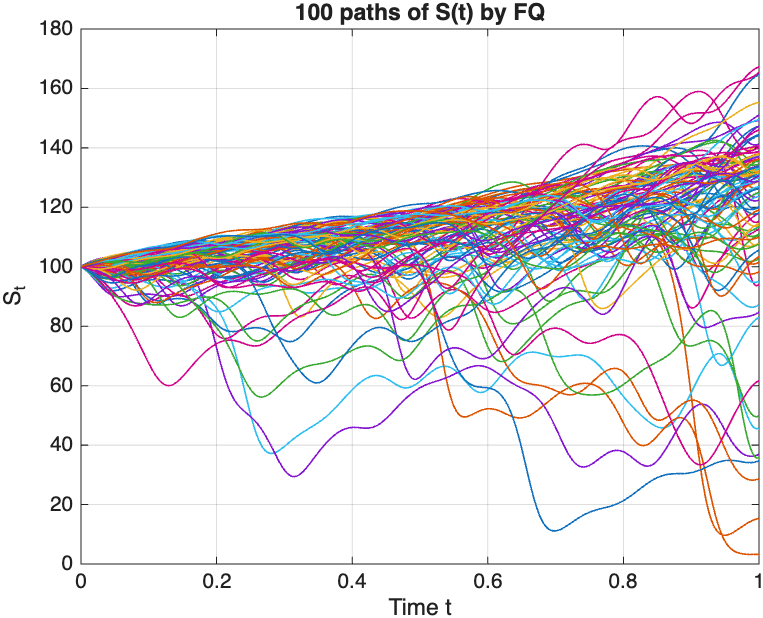}
\end{minipage}
%\hfill
\begin{minipage}{0.45\linewidth}
\centering
\includegraphics[width=\linewidth]{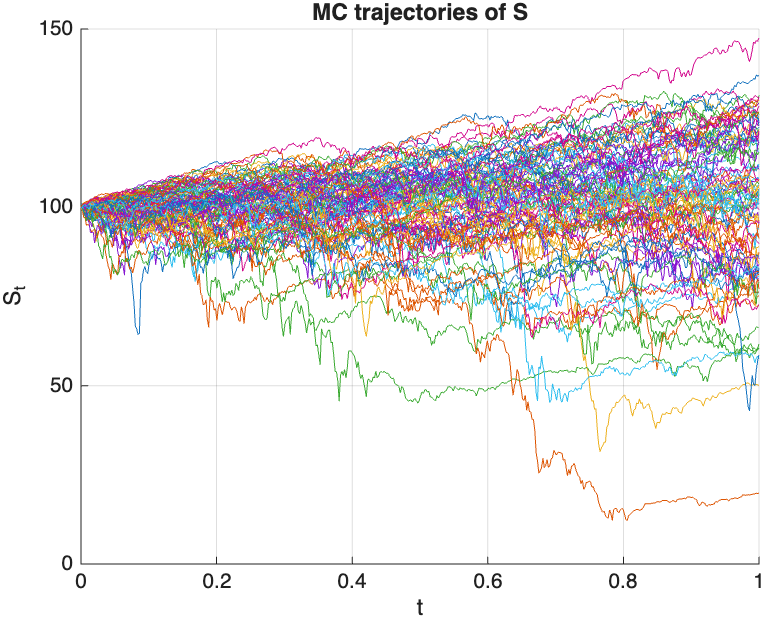}
\end{minipage}

\vspace{0.5cm}

% -------- Seconda riga --------
\begin{minipage}{0.45\linewidth}
\centering
\includegraphics[width=\linewidth]{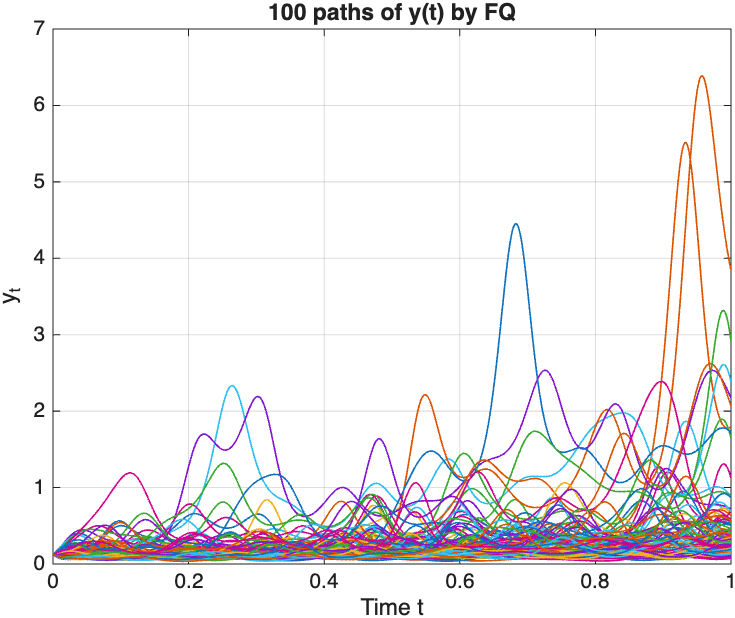}
\end{minipage}
%\hfill
\begin{minipage}{0.45\linewidth}
\centering
\includegraphics[width=\linewidth]{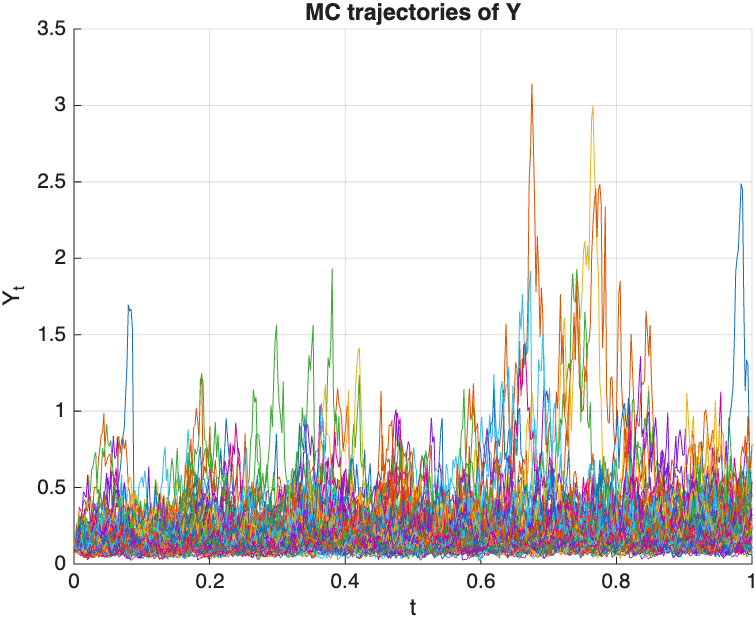}
\end{minipage}

\caption{Visualization of 100 sample paths of the process $S$ (top panels) and the process $y$ (bottom panels) for the 2-factor specification of the model of \cite{GuyonVolMostlyPathDep2022}. The left panels correspond to trajectories obtained via functional quantization, while the right panels display trajectories generated from a Monte Carlo simulation with $10^5$ paths. The model parameters are as follows: $\beta_0=0.08;\ \beta_1=-0.08;\ \beta_2=0.5;\ \lambda_1=62;\ \lambda_2=40;\  R_{1,0}=-0.044;\ R_{2,0}=0.007;\ S_0=100;\ r=0;\ T=1$ with $500$ steps for the time discretization and $
\sigma_0=y_0 = 0.1254$. }
\label{fig:Julien100pathSandY}
\end{figure}

We observe that, under the parametrization proposed in \cite{GuyonVolMostlyPathDep2022}, the dynamics are strongly influenced by the relatively large values of the parameters $\lambda_1$ and $\lambda_2$. These parameters control the strength of the memory effect and enter the exponential kernel defining the volatility process $\sigma$ (hereafter interchangeably denoted by $y$ with a slight abuse of notation). As a consequence, the persistence induced by the kernel significantly affects the temporal evolution of the process, highlighting the central role of long-range dependence in shaping its pathwise behavior.

We therefore carry out a brief sensitivity analysis with respect to variations in the model parameters $\lambda_1, \lambda_2$ and $\beta_1, \beta_2$, with the aim of gaining preliminary insight into their influence, first on the pathwise behavior of the processes and subsequently on the implied volatility smile of the underlying asset.

In Figure~\ref{fig:Julien_impact_lambda1_lambda2_on_y}, we illustrate the impact of variations in $\lambda_1$ and $\lambda_2$ on the mean and standard deviation of the (time-averaged) codeword $y_t$. Figure~\ref{fig:Julien_impact_beta1_beta2_on_y} presents the corresponding sensitivity analysis with respect to the parameters $\beta_1$ and $\beta_2$.

We observe that $\lambda_1$ exerts a markedly stronger influence on the volatility process $\sigma$ than $\lambda_2$. This behavior is expected, as $\lambda_1$ enters the dynamics of $\sigma$ through the stochastic integral term, thereby directly amplifying or damping the standard deviation of volatility. In particular, for sufficiently small values of $\lambda_1$, the volatility process becomes effectively neutralized, driving the model toward a Black-Scholes-type regime with (approximately) constant volatility and, consequently, a flat implied volatility smile.

Turning to the parameters $\beta_1$ and $\beta_2$, whose magnitudes are significantly smaller than those of $\lambda_1$ and $\lambda_2$, the situation is reversed. In this case, $\beta_2$ has a substantially larger impact--exceeding that of $\beta_1$ by more than a factor of two in absolute value--on the level of volatility, highlighting its dominant role within this parameter pair.

% SENSITIVITY of y with FQ  wrt lambda 1 and lambda 2

%\usepackage{graphicx}
%\usepackage{subcaption}

\begin{figure}[!htbp]
\centering

\begin{minipage}{0.24\linewidth}
\centering
\includegraphics[width=\linewidth]{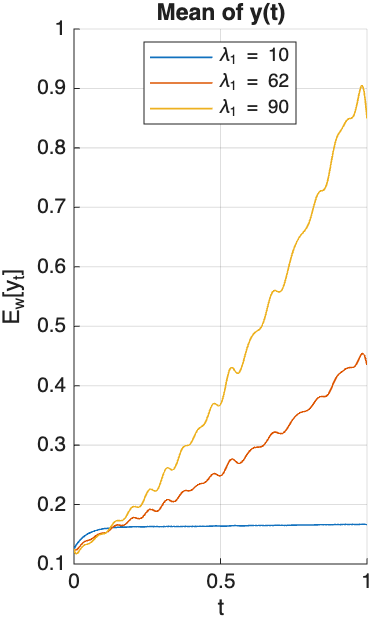}
%\subcaption{A}
\end{minipage}
\hfill
\begin{minipage}{0.24\linewidth}
\centering
\includegraphics[width=\linewidth]{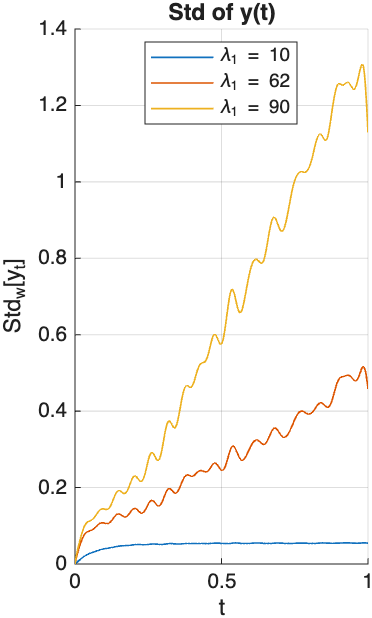}
%\subcaption{B}
\end{minipage}
\hfill
\begin{minipage}{0.24\linewidth}
\centering
\includegraphics[width=\linewidth]{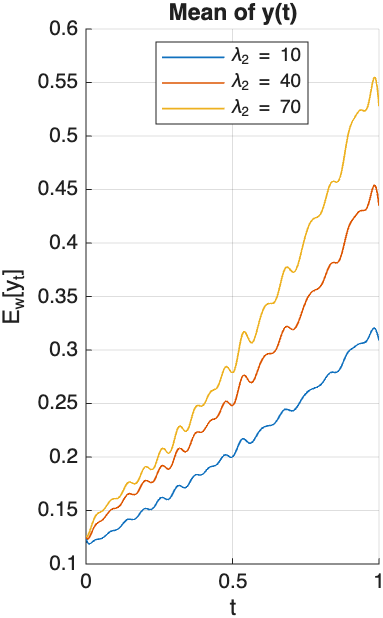}
%\subcaption{C}
\end{minipage}
\hfill
\begin{minipage}{0.24\linewidth}
\centering
\includegraphics[width=\linewidth]{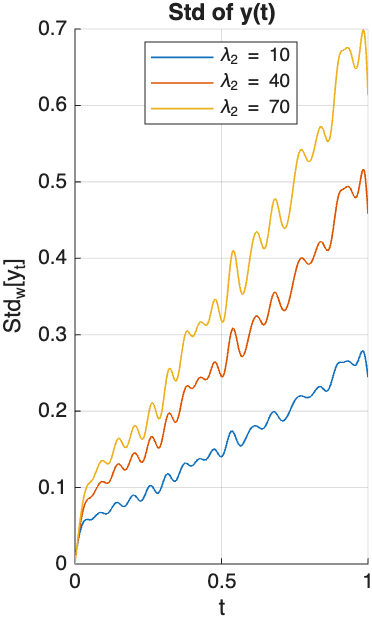}
%\subcaption{D}
\end{minipage}

\caption{Impact of variations in $\lambda_1$ and $\lambda_2$ on the mean and standard deviation of the (time-averaged) codeword $y_t$.}
   \label{fig:Julien_impact_lambda1_lambda2_on_y}
\end{figure}

% SENSITIVITY of y with FQ  wrt beta1 and beta2

%\usepackage{graphicx}
%\usepackage{subcaption}

\begin{figure}[!htbp]
\centering

\begin{minipage}{0.24\linewidth}
\centering
\includegraphics[width=\linewidth]{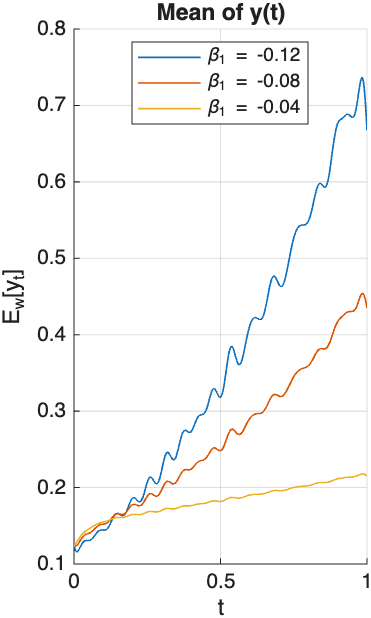}
%\subcaption{A}
\end{minipage}
\hfill
\begin{minipage}{0.24\linewidth}
\centering
\includegraphics[width=\linewidth]{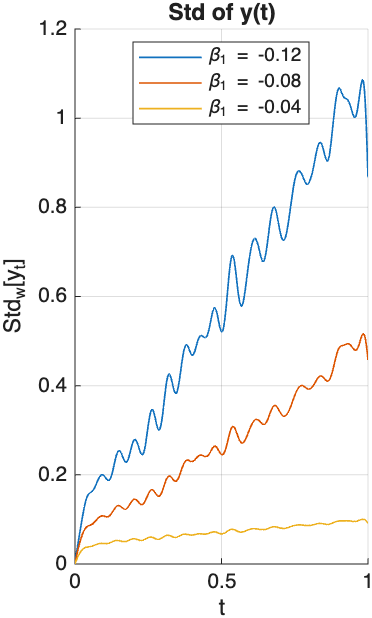}
%\subcaption{B}
\end{minipage}
\hfill
\begin{minipage}{0.24\linewidth}
\centering
\includegraphics[width=\linewidth]{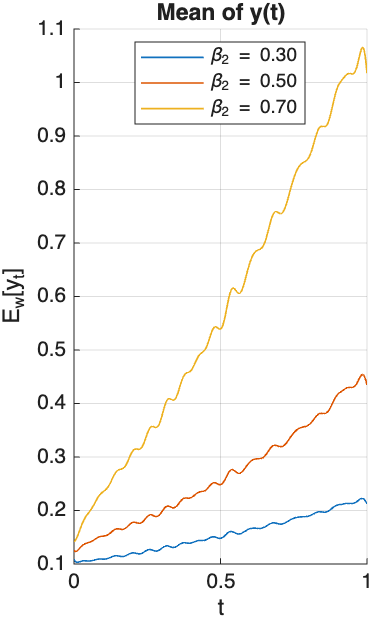}
%\subcaption{C}
\end{minipage}
\hfill
\begin{minipage}{0.24\linewidth}
\centering
\includegraphics[width=\linewidth]{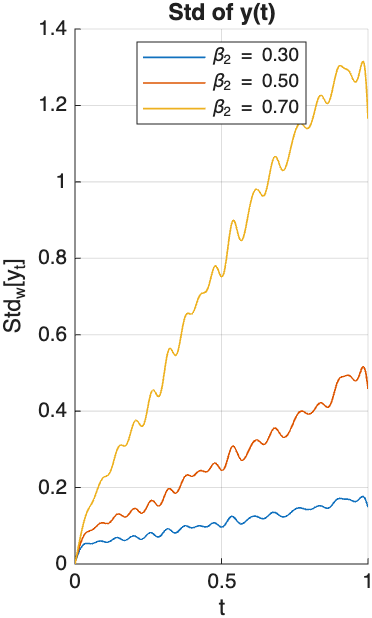}
%\subcaption{D}
\end{minipage}

\caption{Impact of variations in $\beta_1$ and $\beta_2$ on the mean and standard deviation of the (time-averaged) codeword $y_t$.}
   \label{fig:Julien_impact_beta1_beta2_on_y}
\end{figure}

We now turn to assessing the ability of the quantization approach to reproduce the stylized features of the option market, with particular emphasis on implied volatility curves. Figure~\ref{fig:Smiles_julien} compares implied volatility smiles and skews across a range of maturities, from one month up to twelve months, generated using the model parameters via product functional quantization, with those obtained from Monte Carlo simulation. This comparison closely mirrors the analysis presented in \cite{GuyonVolMostlyPathDep2022}. The results demonstrate a good level of agreement: functional quantization accurately captures the shape and term structure of implied volatility, successfully reproducing the key market features and yielding outcomes that are fully consistent with the Monte Carlo benchmark.

We conclude this numerical illustration by quantifying the impact of the Stratonovich correction term. In Figure~\ref{fig:Smiles_julien}, we additionally report the implied volatility smiles obtained when the corrective term is omitted in the functional quantization of the asset price process $S$, while keeping the correctly quantized codewords for the volatility process $y$. The resulting smiles differ substantially from those previously reported, highlighting the presence of a significant bias induced by neglecting the correction. This discrepancy would likely be even more pronounced if the corrective term were also omitted from the ODEs governing the codewords of the process $y$.

% SMILE FQ vs MC vs FQnaive

\begin{figure}[!htbp]
\centering

\begin{minipage}{0.45\linewidth}
\centering
\includegraphics[width=\linewidth]{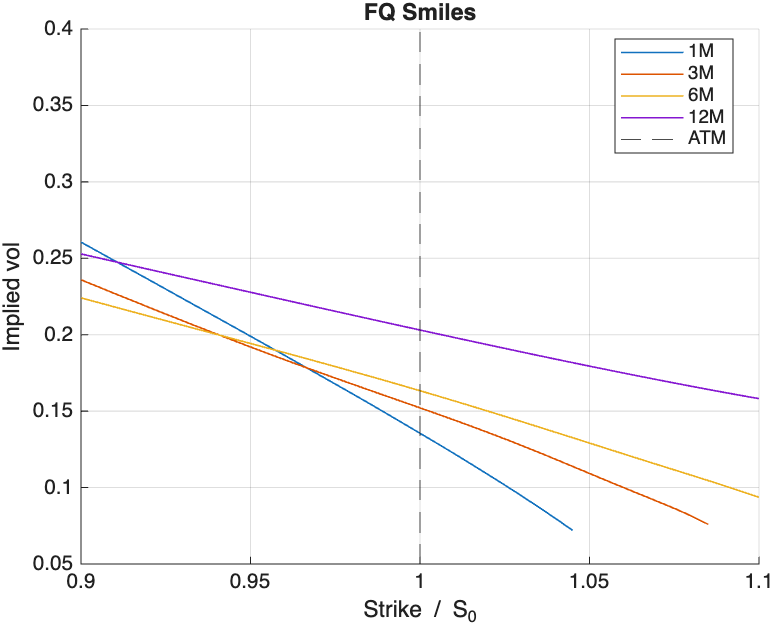}
%\subcaption{A}
\end{minipage}
\hfill
\begin{minipage}{0.45\linewidth}
\centering
\includegraphics[width=\linewidth]{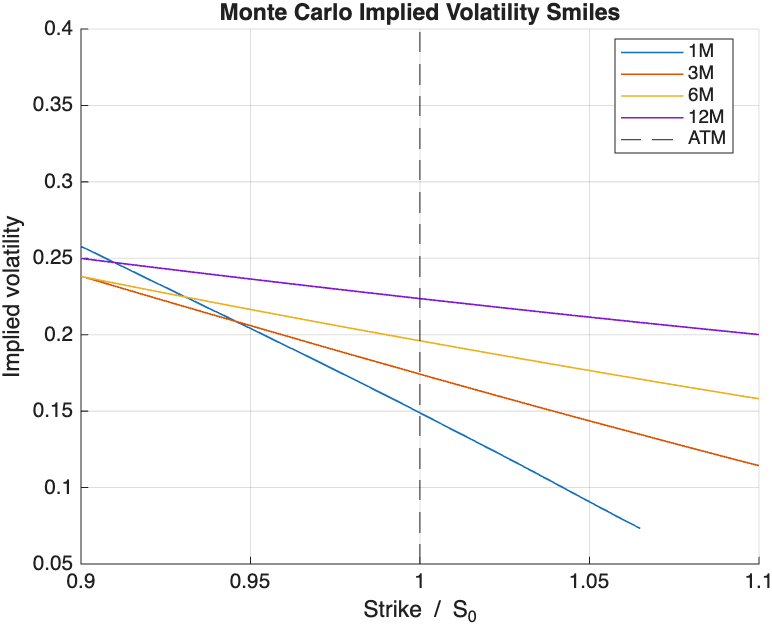}
%\subcaption{B}
\end{minipage}
\hfill
\begin{minipage}{0.45\linewidth}
\centering
\includegraphics[width=\linewidth]{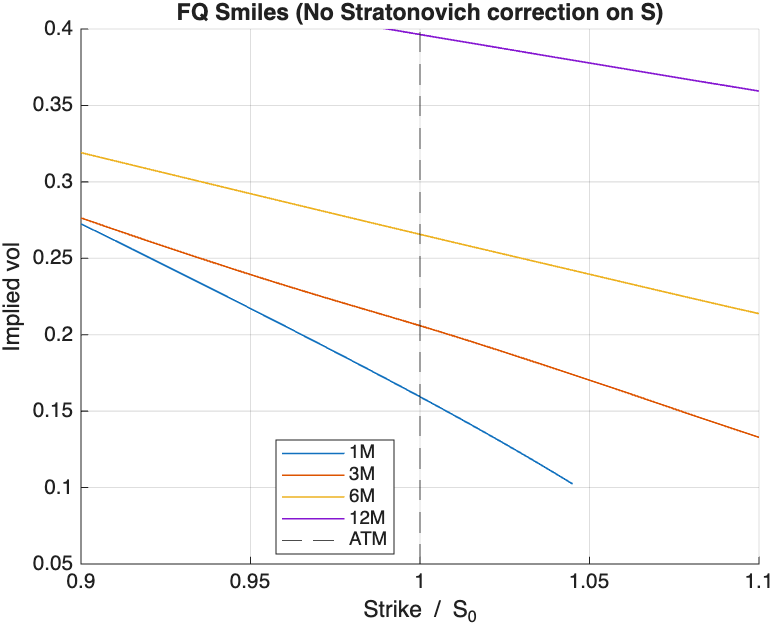}
%\subcaption{C}
\end{minipage}
    
    \caption{Comparison between implied volatility smiles and skews across a range of maturities, from one month up to twelve months, generated using the model parameters via product functional quantization ("QF smiles" on the top left of the panel), with those obtained from Monte Carlo simulation ("Monte Carlo Smiles", on the top right of the panel) and those obtained when the corrective term is omitted in the functional quantization of the asset price process $S$, while keeping the correctly quantized codewords for the volatility process $y$ ("FQ Smiles -- No Stratonovich correction on S" on the bottom of the panel. }
    \label{fig:Smiles_julien}
\end{figure}

\section{Application to the Model of  \cite{platrendek18}}\label{section5}

\subsection{Existence and Uniqueness  for the Volatility Process}

We now address the existence and uniqueness of strong solutions for the stochastic differential equation (SDE) defining the process~\eqref{Platen1} in the model of \cite{platrendek18}. Our approach follows the lines of Exercise 34 p.130 in   \cite{Lamberton97} and  leverages again Theorem 11.2 from \cite{rogerswill2000}, which has proven to be a powerful tool in establishing such results for SDE systems of the form~\eqref{R1dyn} and~\eqref{R2dyn}.

First of all fix an initial condition $y_0>0$ and parameters $\epsilon, A$ such that $0<\epsilon <y_0<A$. We define the truncation operator $T_{\epsilon, A}(y): t\rightarrow (y(t)\land A)\lor  \epsilon$ and 
\begin{align*}
b_{\epsilon, A}(t,y)&=b\left(t,T_{\epsilon, A}(y)_t,\int_0^te^{\lambda u}\sqrt{T_{\epsilon, A}(y)_u}du,0\right)\\
a_{\epsilon, A}(t,y)&=a\left(T_{\epsilon, A}(y)_t,\int_0^te^{\lambda u}\sqrt{T_{\epsilon, A}(y)_u}du\right).
\end{align*}
Consider the following functional SDE for $t\in[0,T],\ T>0$:
\begin{align}
Y_t^{\epsilon, A}&=y_0+\int_0^t b_{\epsilon, A}(s,Y_s^{\epsilon, A})ds+
\int_0^t a_{\epsilon, A}(s,Y_s^{\epsilon, A})dW_s.\label{functionalYPlaten}
\end{align}
As the map $u\rightarrow \sqrt{u}$ is $\frac{1}{\sqrt{\epsilon}}-$ Lipschitz on $[\epsilon, A]$, it is easy to check that the coefficients of~\eqref{functionalYPlaten} are locally Lipschitz, namely
\begin{align*}
\vert  b_{\epsilon, A}(t,x)-b_{\epsilon, A}(t,y)\vert +\vert a_{\epsilon, A}(t,x)-a_{\epsilon, A}(t,y) \vert
&\leq C_{\epsilon, A,T,\eta} \vert\vert x-y\|_{[0,T]}.
\end{align*}
We can then apply Theorem 11.2 in \cite{rogerswill2000} which proves the existence of a unique solution on $\mathbb{R}_+$ to the Equation~\eqref{functionalYPlaten}.
Now let introduce 
\begin{align*}
\tau_{\epsilon, A}:=\inf \{ t: Y_t^{\epsilon, A} \in \{ \epsilon, A\}\}>0\ \ \mbox{as}\ y_0\in ]\epsilon, A[.
\end{align*}
From the uniqueness of the solution, it is clear that if $0<\epsilon'<\epsilon <y_0<A<A'$ then 
$Y_{\tau_{\epsilon, A}}^{\epsilon', A'}=Y_{\tau_{\epsilon, A}}^{\epsilon, A}$.
In other terms, the solutions are telescopic, so that  one can define a process $\left( Y_t\right)_{t\in [0,\tau[}$, with $\tau=\sup_{0<\epsilon<A}\tau_{\epsilon, A}\leq +\infty$, where for all $t\in[0,\tau[$ (i.e. on $\{\tau>t\}$) we have
\begin{align}
Y_t&=y_0+\int_0^t b_{\epsilon, A}(s,Y_s)ds+
\int_0^t a_{\epsilon, A}(s,Y_s)dW_s.\label{functionalYPlaten2}
\end{align}
The next step is to show that $\tau =+\infty\  \mathbb{P}-a.s.$, in order to extend the solution to the entire positive real time line. To this aim, we need to have a closer look at the original model~\eqref{Platen1}:
\begin{align*}
  dY_t^{y_0}  &= (\alpha- \beta Y_t^{y_0}) M_t dt  + \sigma \sqrt{M_t Y_t^{y_0}}  dW_t, \qquad Y_0=y_0>0,\\
M_t & = \xi (\lambda^2 (2 \sqrt{Y_t}-Z_t)^2 + \eta  )\geq \xi\eta>0.
\end{align*}
Let consider the scale function
% (see e.g. Lamberton and Lapeyre ???? pag. 137, Exercise 34)
\begin{align*}
s(x)=\int_1^x e^{\frac{2\beta}{\sigma^2}u}u^{-\frac{2\alpha}{\sigma^2}} du,
\end{align*}
which verifies the following ODE
\begin{align*}
\frac{\sigma^2}{2}xs"(x)+(\alpha -\beta x)s'(x)=0.
\end{align*}
Note also that we can rewrite 
$\tau_{\epsilon, A}=\tau_{\epsilon}\land \tau_{ A}$, where for any  $ c\in \{\epsilon, A\}$ we have
\begin{align*}
\tau_c&=\inf \{ t>0: Y_t=c\}=\inf \{ t>0: Y_t^{\epsilon',A'}=c\},
\end{align*}
where $0<\epsilon'<\epsilon <y_0<A<A'$.
Therefore,
\begin{align}
s(Y_{t \land \tau_{\epsilon, A}})&=
\int_0^{t \land \tau_{\epsilon, A} }
\left(\left(\alpha -b Y_u\right) s'(Y_u) + \frac{\sigma^2}{2}Y_u s"(Y_s)\right)
M_u du+
\sigma \int_0^{t \land \tau_{\epsilon, A} } s'(Y_u)\sqrt{M_uY_u}dW_u\nonumber\\
&= 
\sigma \int_0^{t \land \tau_{\epsilon, A} } s'(Y_u^{\epsilon, A})\sqrt{M_u^{\epsilon, A}Y_u^{\epsilon, A}}dW_u,\label{sMartingale}
\end{align}
which is a true martingale  since the argument of the stochastic integral  is bounded  by a deterministic function for any $\alpha\geq \sigma^2/2$ (where we used  an obvious notation for $M_u^{\epsilon, A}$).
This proves that $s(Y_{t \land \tau_{\epsilon, A}})_{t\geq 0}$ is a true martingale on $\mathbb{R}_+$ (recall that $ \tau_{\epsilon, A}<\tau$)  centered on $s(y_0)$. On the other hand, one can verify that the \textcolor{black}{second moment} satisfies the following inequality
\begin{align*}
\mathbb{E}\left[s(Y^{y_0}_{t \land \tau_{\epsilon, A}})^2\right]&=
\mathbb{E}\left[\sigma^2 \int_0^{t \land \tau_{\epsilon, A} } s'(Y_u^{\epsilon, A})^2 M_u^{\epsilon, A}Y_u^{\epsilon, A}du\right]\geq \sigma^2 \xi\epsilon \eta \inf_{x\in [\epsilon, A]}(s'(x))^2 
\mathbb{E}\left[t \land \tau_{\epsilon, A}\right],
\end{align*}
where we note that $\inf_{x\in [\epsilon, A]}(s'(x))^2>0$. Now,
\begin{align*}
\mathbb{E}\left[t \land \tau_{\epsilon, A}\right]\leq
\frac{\mathbb{E}\left[s(Y^{y_0}_{t \land \tau_{\epsilon, A}})^2\right]}{ \sigma^2 \xi\epsilon \eta \inf_{x\in [\epsilon, A]}(s'(x))^2 
}\leq
\frac{C_{y_0,\epsilon, A, \eta}}{ \sigma^2 \xi\epsilon \eta \inf_{x\in [\epsilon, A]}(s'(x))^2 
},
\end{align*}
where the last inequality follows from the boundedness of $s(.)$ on $[\epsilon, A]$.
From Fatou's lemma, we get 
\begin{align*}
\mathbb{E}[\tau_{\epsilon, A}]\leq
\lim_{t\rightarrow +\infty}\mathbb{E}[t \land \tau_{\epsilon, A}]
\leq
\frac{C_{y_0,\epsilon, A, \eta}}{ \sigma^2 \xi\epsilon \eta \inf_{x\in [\epsilon, A]}(s'(x))^2 
}<+\infty,
\end{align*}
i.e. $ \tau_{\epsilon, A}<+\infty$ $\mathbb{P}-a.s.$
We now take the expected value in~\eqref{sMartingale} and for $t\rightarrow +\infty$ we get
\begin{align*}
s(y_0)=&\lim_{t\rightarrow +\infty}
\textcolor{black}{
\mathbb{E}\left[s(Y^{y_0}_{t \land \tau_{\epsilon, A} })\right]=
\mathbb{E}\left[s(Y^{y_0}_{ \tau_{\epsilon, A} })\right],
}
\end{align*}
i.e. 
\begin{align}
s(y_0)=&s(\epsilon)\mathbb{P}(\tau_\epsilon<\tau_A)+s(A)\mathbb{P}(\tau_\epsilon>\tau_A).\label{sy0}
\end{align}
Note that under the Feller assumption $\alpha\geq \sigma^2/2$ we have $\lim_{x\rightarrow 0}s(x)=-\infty$, and $\tau_\epsilon \uparrow \tau_0:=\inf\{t:Y_t=0\}$ as $\epsilon \downarrow 0$, so that 
$\mathbb{P}(\tau_0<\tau_A)\leq \mathbb{P}(\tau_\epsilon<\tau_A)$ for all $\epsilon \in]0,y_0[$. Moreover, under the (reasonable) assumption that $0<\epsilon<1$, we have $s(\epsilon)<0$, so that
\begin{align*}
\mathbb{P}(\tau_\epsilon<\tau_A)&=\frac{s(y_0)}{s(\epsilon)}+\frac{s(A)}{(-s(\epsilon))}\mathbb{P}(\tau_\epsilon>\tau_A)\leq \frac{\vert s(y_0)\vert} {\vert s(\epsilon)\vert }+\frac{\vert s(A)\vert }{\vert s(\epsilon)\vert }\times 1.
\end{align*}
Now taking $\epsilon \rightarrow 0$, it follows that
\begin{align*}
\mathbb{P}(\tau_0<\tau_A)&=\lim_{\epsilon \rightarrow 0}\mathbb{P}(\tau_\epsilon<\tau_A)\leq \frac{\vert s(y_0)\vert+\vert s(A)\vert }{+\infty }=0,
\end{align*}
i.e. $\mathbb{P}(\tau_0<\tau_A)=0$ $  \forall A>0$, therefore  $\tau_0\geq \tau_A$ $\mathbb{P}-a.s.$, then $Y_t>0$ on $[0,\tau_A[$ $\forall A>0$, i.e. $Y_t=Y^{y_0}_{t \land \tau_{ A} }>0$ on $\cup\uparrow _{A>0}[0,\tau_A[=[0,\tau[$.
We have only to conclude for $\tau_A$. We restart from~\eqref{sy0} and we note that for $\epsilon<y_0 \land 1$
\begin{align*}
0\leq \mathbb{P}(\tau_A<\tau_\epsilon)= \frac{ s(y_0)+(-s(\epsilon))\mathbb{P}(\tau_\epsilon<\tau_A)}{s(A)}\leq \frac{ s(y_0)-s(\epsilon))}{s(A)}, 
\end{align*}
i.e.  (note that $\beta\geq 0$)
\begin{align*}
\lim_{A \rightarrow +\infty}\mathbb{P}(\tau_A<\tau_\epsilon)=\frac{ s(y_0)-s(\epsilon))}{+\infty}=0,
\end{align*}
and, as $\tau=\sup_{A>0}\tau_A$, 
\begin{align*}
\mathbb{P}(\tau<\tau_\epsilon)=0,\ \forall \epsilon \in [0,y_0\land 1[.
\end{align*}
Now since  $\tau_\epsilon \uparrow +\infty$ for $\epsilon \downarrow 0$, we get 
\begin{align*}
\mathbb{P}(\tau<+\infty)=\lim_{\epsilon \rightarrow 0}\mathbb{P}(\tau<\tau_\epsilon)=0,
\end{align*}
 that is $\tau=+\infty \ \mathbb{P}-a.s.$, which ensures the existence (and also the uniqueness) of a strong solution of the volatility process of the model.

\subsection{Lamperti Transform in the Model of \cite{platrendek18}}

From~\eqref{aPlaten}, the Lamperti transform~\eqref{lampertiytildey}  in the model of \cite{platrendek18} is given by
\begin{align*}
S(t,y,\tilde{y}^h) &= \int_{\epsilon_0}^y \frac{dx}{a(t,x, \tilde{y}^h)}\\
&= \int_{\epsilon_0}^y \frac{dx}{\sigma \sqrt{\xi x} \sqrt{ 4\lambda^2 \left( \sqrt{x} - \lambda e^{-\lambda t} \tilde{y}^h \right)^2 + \eta}}\\
&=\textcolor{black}{\left[    
\frac{1}{\sigma \lambda\sqrt{\xi}} 
  \text{ ArcTanh}\left(\frac{2 \lambda \left( \sqrt{x} - \lambda e^{-\lambda t} \tilde{y}^h \right)}
  {\sqrt{ 4\lambda^2 \left( \sqrt{x} - \lambda e^{-\lambda t} \tilde{y}^h \right)^2 + \eta}}
\right)
\right]_{\epsilon_0}^y,}
\end{align*}
which is expressed in terms of the  inverse hyperbolic tangent function. 
Taking the inverse w.r.t.  $y$, gives 
\textcolor{black}{\begin{align*}
S^{-1}_y(t,x,\tilde{y}^h)&=
\left(
\lambda e^{-\lambda t}(\tilde y^h)^2+\frac{\sqrt{\eta} \text{Tanh} \left(\frac{x+c}{\sigma\lambda \sqrt{\xi}}\right)}{2\lambda \sqrt{1-\text{Tanh}^2 \left(\frac{x+c}{\sigma\lambda \sqrt{\xi}}\right)}}\right)^2,
\end{align*}}
\textcolor{black}{where we denoted with  $c$ the previous integral evaluated at the lower extremal $\epsilon_0$.}

% CODE FOR COMPUTING DERIVATIVE OF S wrt \tilde y=z

%f[x_, t_, z_] := -((2 σ Sqrt[ξ])/λ) ArcTanh[
%  (2 λ Sqrt[x])/
%   (Sqrt[η + 4 Exp[-2 λ t] z^2 λ^4] - 
%    Sqrt[η + 4 Exp[-2 λ t] λ^2 (Exp[λ t] Sqrt[x] - z λ)^2])
%]
%
%derivata = D[f[x, t, z], z] // FullSimplify
%
%Print[derivata]

Now let consider the derivative of the Lamperti transform w.r.t. the variables $\tilde{y}^h$ and $t$:
 \textcolor{black}{
 \begin{align*}
\frac{\partial S}{\partial \tilde y^h}(t,y,\tilde y^h)&=
-\frac{2\lambda e^{-\lambda t}}{\sigma \sqrt{\xi}}
\left(
\frac{1}{\sqrt{4\lambda^2 \left( \sqrt{y} - \lambda e^{-\lambda t} \tilde{y}^h \right)^2 + \eta}}
-
\frac{1}{\sqrt{4\lambda^2 \left( \sqrt{\epsilon_0} - \lambda e^{-\lambda t} \tilde{y}^h \right)^2 + \eta}}
\right)\\
\frac{\partial S}{\partial t}(t,y,\tilde y^h)&=
-2\lambda^2 e^{-\lambda t}(\tilde y^h)^2\left(
\lambda e^{-\lambda t}(\tilde y^h)^2+\frac{\sqrt{\eta} \text{Tanh} \left(\frac{x+c}{\sigma\lambda \sqrt{\xi}}\right)}{2\lambda \sqrt{1-\text{Tanh}^2 \left(\frac{x+c}{\sigma\lambda \sqrt{\xi}}\right)}}\right).
\end{align*}
}
We have then all the ingredients in order to write the  ODE satisfied by the functional codewords $x_t$ of the quantizer for the process $X_t$:
 \begin{align*}
d x_t=&\left(
\frac{b(t,S^{-1}_y\textcolor{black}{(t,x,\tilde{y}^h)},\tilde y_t^{g_1},\tilde y_t^{g_2})}{
a (t,S^{-1}_y\textcolor{black}{(t,x,\tilde{y}^h)}, \tilde y^h_t)}
+\frac{\partial S}{\partial \tilde y^h}(t,S^{-1}_y\textcolor{black}{(t,x,\tilde{y}^h)},\tilde y^h_t)h(t,S^{-1}_y\textcolor{black}{(t,x,\tilde{y}^h)})\right. \\
 &\left. 
+\frac{\partial S}{\partial t}(t,S^{-1}_y\textcolor{black}{(t,x,\tilde{y}^h)},\tilde y^h_t)
-\frac{1}{2} \frac{\partial a}{\partial y}(t,S^{-1}_y\textcolor{black}{(t,x,\tilde{y}^h)}, \tilde y^h_t)\right)dt
+\alpha'(t)dt,
\end{align*}
where $a,b$ are given by \eqref{aPlaten}, \eqref{bPlaten} and (recall that here $g_2=0$)
\begin{align*}
\tilde y_t^{g_1}&=\int_{0}^t g_1(u,y_u)du=\int_{0}^t e^{\lambda u} \sqrt{y_u}du =\int_{0}^t h(u,y_u)du=\tilde y_t^{h},\\
\tilde y_t^{g_2}&=0.
\end{align*}

\subsection{ODEs for the Functional Codewords of the Quantizer for the Process $Y$}

From~\eqref{aPlaten}, we have 
\begin{align}
\frac{\partial a}{\partial y}(t,y, \tilde y^h)&=
\frac{\sigma \sqrt{\xi}}{2\sqrt{y}}\left(
\frac{\eta + 8y\lambda^2 -12\lambda^3e^{-\lambda t}\sqrt{y}\tilde y^h + 4\lambda^4 e^{-2\lambda t}(\tilde y^h)^2}{\sqrt{ 4\lambda^2 \left( \sqrt{y} - \lambda e^{-\lambda t} \tilde{y}^h \right)^2 + \eta}}\right),
\end{align}
so that
\begin{align}
a (t,y, \tilde y^h)\frac{\partial a}{\partial y}(t,y, \tilde y^h)&=
\frac{\sigma^2\xi}{2} \left(\eta + 8y\lambda^2 -12\lambda^3e^{-\lambda t}\sqrt{y}\tilde y^h + 4\lambda^4 e^{-2\lambda t}(\tilde y^h)^2)\right).
\end{align}

The ODE satisfied by the functional codewords $y_t$ of the quantizer for the process $Y$ is then 
\begin{align*}
d y_t&= \left(b(t,y_t,\tilde y_t^{g_1},\tilde y_t^{g_2})-\frac{1}{2}a (t,y_t, \tilde y_t^h)\frac{\partial a}{\partial y}(t,y, \tilde y^h_t)\right)dt
+a (t,y_t, \tilde y_t^h)\alpha'(t)dt\\
&=  \xi (\alpha - \beta y_t) \left(4 \lambda^2 \left( \sqrt{y_t} - \lambda e^{-\lambda t} \tilde y_t^{g_1} \right)^2 + \eta \right)dt\\
& \quad -\frac{\sigma^2\xi}{4} \left(\eta + 8y_t\lambda^2 -12\lambda^3e^{-\lambda t}\sqrt{y_t}\tilde y_t^h + 4\lambda^4 e^{-2\lambda t}(\tilde y_t^h)^2\right)dt\\
&\quad +\sigma \sqrt{\xi y_t} \sqrt{ 4\lambda^2 \left( \sqrt{y_t} - \lambda e^{-\lambda t} \tilde{y_t}^h \right)^2 + \eta}
\alpha'(t)dt.
\end{align*}
%where we recall that $\tilde g_1(t,y_t) =\tilde  h(t,y_t) =\int_0^t \sqrt{y_u} e^{\lambda u}du$.

\subsection{Functional Quantization for the Growth Optimal Portfolio Process}

From~\eqref{modelPlaten} we get
\[
S_t  = S_0 \exp\left( \int_0^t \Big(r_s + \frac{M_s}{  2 Y_s}\Big) ds  +   \int_0^t  \sqrt{\frac{M_s}{Y_s}} dW_s  \right).
\]
Now we apply the product functional quantization as in \eqref{sellami} to  the stochastic integral 
\begin{align*}
\tilde Y^{f}_t :
= \int_0^t  \sqrt{\frac{M_s}{Y_s}} dW_s= \int_0^t f\left(s,Y_s,\tilde Y^h_s\right)  dW_s,
\end{align*}
where $\tilde Y^h_t=\int_0^t h(u,Y_u)du$ with $h(t,y)$ given in \eqref{hPlaten} (i.e.,  $h(t,y)=\sqrt{y}e^{\lambda t}$)
and 
\begin{align}
f\left(t,Y_t,\tilde Y^h_t \right)&= \sqrt{\frac{M_t}{Y_t}}=  \sqrt{\frac{\xi  \left( 4\lambda^2 \left( \sqrt{Y_t} - \lambda e^{-\lambda t} \tilde Y^h_t \right)^2 + \eta \right)}{Y_t}}.\label{fPlaten}
\end{align}
From 
\begin{align*}
\tilde Y^{f}_t = \int_0^t f\left(u,Y_u,\tilde{Y}_u^h\right) \circ dW_u
 - \frac{1}{2}\int_0^t a(u,Y_u, \tilde{Y}_u^h)\frac{\partial f}{\partial y}\left(u,Y_u,\tilde{Y}_u^h\right)du,\end{align*}
 we apply \eqref{sellami} to
the Stratonovich  integral and it follows that   the functional codeword of the above  It\^o stochastic integral $\tilde Y^{f}_t$ is given  by
\begin{align*}
\tilde y^{f}_t = \int_0^t f\left(u,y_u,\tilde y^h_u\right)\alpha'(u)du- \frac{1}{2}\int_0^t a (u,y_u,\tilde{y}_u^{h})\frac{\partial f}{\partial  y}\left(u,y_u,\tilde y^h_u\right)du,
\end{align*}
where $a$ is given by \eqref{aPlaten} and $\alpha$ is the usual codeword of a quantizer of $W$.
Finally, we can approximate the GOP process $S_t$, for $0=s_0<s_1<\cdots <s_n=t$ ($s_i=i\Delta; i=0,\cdots, n; \Delta=t/n$), for any $n\ge 1$, by 
{\small 
\begin{eqnarray}
\hat S_t  &=& S_0 \exp\left( \Delta  \sum_{i=1}^{n} \left( \Big(r_{s_i} + \frac{\hat {\bar M}_{s_i}}{  2 \hat {\bar Y}_{s_i}}
- \frac{1}{2} a (s_i,\hat {\bar Y}_{s_i},\hat {\bar {\tilde Y}}_{s_i}^{h})\frac{\partial f}{\partial  y}\left(s_i,\hat {\bar Y}_{s_i},\hat {\bar {\tilde Y}}_{s_i}^{h}\right)
\Big)  \right.\right. \nonumber\\
&& \left.\left. +   \sum_{\ell=1}^{d_N}  \sqrt{\frac{2}{t}}\ \hat\xi_{\ell}^{N_{\ell}}  \cos \left( \frac{s_i}{\sqrt{\lambda_{\ell}}}\right) \sqrt{\hat {\bar M}_{s_i} /  \hat {\bar Y}_{s_i}} \right) \right) \nonumber\\
& = & S_0 \exp\left( \Delta  \sum_{i=1}^{n} \left( \Big(r_{s_i} + \frac{\hat {\bar M}_{s_i}}{  2 \hat {\bar Y}_{s_i}}
- \frac{1}{2} a (\hat {\bar Y}_{s_i},\hat {\bar {\tilde Y}}_{s_i}^{h})\frac{\partial f}{\partial  y}\left(s_i,\hat {\bar Y}_{s_i},\hat {\bar {\tilde Y}}_{s_i}^{h}\right)
\Big)   +     \sqrt{\frac{2}{t}}\ \sqrt{\hat {\bar M}_{s_i} /  \hat {\bar Y}_{s_i}} \  \hat \xi_{s_i} \right) \right)
 \label{GOPquantized}\end{eqnarray}
}
where
\begin{eqnarray*}
\hat \xi_{s_i} & =& \sum_{\ell=1}^{d_N} \hat\xi_{\ell}^{N_{\ell}} \cos \left( \frac{s_i}{\sqrt{\lambda_{\ell}}}\right), \\
 \hat {\bar M}_{s_i} & =& \xi \Big(\lambda^2 \Big(2 \hat {\bar Y}_{s_i}^{^{1/2}}-  2 \lambda  \sum_{j=0}^i  e^{-\lambda(t-s_j)}  \hat{\bar Y}_{s_j}^{^{1 / 2}}   \Big)^2 + \eta  \Big),\\
\hat {\bar {\tilde Y}}_{s_i}^{h}&=&\sum_{j=0}^i  e^{\lambda s_j}  \hat{\bar Y}_{s_j}^{^{1 / 2}}.
\end{eqnarray*}
Finally,  from \eqref{aPlaten} and \eqref{fPlaten}, it follows that 
\begin{align*}
 a (t,y_t,\tilde{y}_t^{h})\frac{\partial f}{\partial  y}\left(t,y_t,\tilde y^h_t\right)=
 \frac{\sigma \xi}{2y_t}\left(
 4\lambda^3 e^{-\lambda t}\sqrt{y_t}\tilde y^h_t-4\lambda^4 e^{-2\lambda t}(\tilde y^h_t)^2-\eta\right).
\end{align*}

\subsection{Pricing a Zero Coupon Bond Under the Benchmark Approach}

From Formula \eqref{benchmarkZCB} giving the price of a ZCB under the benchmark approach, it follows that we need to quantize the inverse of the GOP. From \eqref{GOPquantized} we get that $\hat S_t^{-1}$ can be approximated recursively by
\[
\hat S_0^{-1} = S_0^{-1}, \qquad \hat S_{s_{i}} ^{-1}= \hat S_{s_{i-1}}^{-1}   \exp\Big(- \Delta   \textcolor{black}{\left( \mbox{drift}_{s_i}   +     \sqrt{2 / t}\    \hat \xi_{s_i}\mbox{vol}_{s_i}  \right)} \Big)
\]
where  for $i=0,\cdots, n$:
\begin{eqnarray*}
\mbox{drift}_{s_i }&=&
 r_{s_i } + \frac{\hat {\bar M}_{s_i}}{  2 \hat {\bar Y}_{s_i}}
- \frac{1}{2} a (\hat {\bar Y}_{s_i},\hat {\bar {\tilde Y}}_{s_i}^{h})\frac{\partial f}{\partial  y}\left(s_i,\hat {\bar Y}_{s_i},\hat {\bar {\tilde Y}}_{s_i}^{h}\right)\\
&=& r_{s_i } + \frac{\hat {\bar M}_{s_i}}{  2 \hat {\bar Y}_{s_i}}
-  \frac{\sigma \xi}{4\hat {\bar Y}_{s_i}}\left(
 4\lambda^3 e^{-\lambda s_i}\sqrt{\hat {\bar Y}_{s_i}}\hat {\bar {\tilde Y}}_{s_i}^{h}-4\lambda^4 e^{-2\lambda s_i}(\hat {\bar {\tilde Y}}_{s_i}^{h})^2-\eta\right).
\\
 \mbox{vol}_{s_i }& =& \sqrt{\hat {\bar M}_{s_i} /  \hat {\bar Y}_{s_i}}
\end{eqnarray*}
  We may also write 
\begin{equation}
\hat S_t^{-1} =  \hat S_{s_{n}}^{-1} = S_0^{-1} \prod_{i=1}^n    \exp\Big(- \Delta   \big( \mbox{drift}_{s_i}    +     \sqrt{2 / t}\ \hat \xi_{s_i} \  \mbox{vol}_{s_i}  \big) \Big).
\end{equation}
Since for any $s \ge 0$, $\hat {\bar M}_s$ and $\hat{\bar Y}_s$ are functions of the random variable $\hat \xi_s$, we deduce that $\hat S_t^{-1}$ can be written as a function of the random variable $\hat \xi:= \sum_{\ell=1}^{d_N} \hat\xi_{\ell}^{N_{\ell}}$: 
\[
\hat S_t^{-1} = \Psi(\hat \xi), \ \mbox{ with } \   \Psi(\mathbf{x}) := S_0^{-1} \prod_{i=1}^n    \exp\Big(- \Delta   \big( \mbox{drift}_{s_i}    +     \sqrt{2 / t}\ \textcolor{black}{ \mathbf{ x}_s\mathds{1}_{ s=s_i}} \ \cos \left( \frac{s_i}{\sqrt{\lambda_{\ell}}}\right)  \  \mbox{vol}_{s_i}  \big) \Big).
\]
It follows that the (product functional quantization based) price of the zero coupon bond is given by
{\small 
\begin{eqnarray}
\mathbb E \big(\hat S_t^{-1} \big)   &=& S_0\sum_{j=1}^N   \Psi\Big(\sum_{\ell=1}^{d_N} x_{j_{\ell}} \Big)\  \mathbb P\Big(\hat \xi = \sum_{\ell=1}^{d_N} x_{j_{\ell}}\Big) \nonumber\\
&=&     \sum_{j=1}^N  \prod_{i=1}^n    \exp\Big(- \Delta   \big( \mbox{drift}_{s_i}    +     \sqrt{2 / t}\  \mbox{vol}_{s_i}   \sum_{\ell=1}^{d_N} x_{j_{\ell}}  \big) \Big) \prod_{\ell=1}^{d_N} \big( \Phi_0\big(x^{\ell}_{j_{\ell}+} \big)  - \Phi_0\big(x^{\ell}_{j_{\ell}-} \big) \big),\label{ZCBond}
\end{eqnarray}
}
where (see e.g. \cite{pagesprintem05}) $\Phi_0(.)$  is the cdf of the ${\cal N}(0; 1)$ and, for $i_\ell =1,\dots ,N_\ell$,
\begin{equation*}
x^{\ell}_{j_{\ell}+} = \frac{x^{\ell}_{j_{\ell}}+x^{\ell}_{j_{\ell}+1}}{2};  \  \ 
x^{\ell}_{j_{\ell}- }= \frac{x^{\ell}_{j_{\ell}}+x^{\ell}_{j_{\ell}-1}}{2}\ \  \hbox{with}\ \ 
x^\ell_{1-} = -\infty, \  \  x^\ell_{N_\ell +} = +\infty.
\end{equation*}
%
%
%}
%\textcolor{black}{

\subsection{Numerical Illustration for the Model of \cite{platrendek18}}

In this subsection, we provide a numerical illustration of the product functional quantization framework applied to the model introduced in \cite{platrendek18}, together with a set of numerical experiments related to the pricing of simple financial instruments, namely a zero-coupon bond and a call option written on the GOP. A comprehensive numerical assessment of the full methodology is beyond the scope of the present work and is therefore deferred to future research.

Let us consider the SDEs \eqref{modelPlaten} and \eqref{Platen1}, which describe the stochastic dynamics of the GOP. In order to  simulating the process $Z$, we adopt the following equivalent SDE formulation,
\begin{equation*}
dZ_t=-\lambda Z_tdt +2\lambda \sqrt{Y_t}dt,\ Z_0=0,
\end{equation*}
which proves to be numerically more stable. Indeed, the model is prone to potential instabilities and therefore requires particular care in the choice of the parameter values.

Indeed, in order to achieve a satisfactory fit of the monthly and daily dynamics observed for the S\&P 500,  \cite{platrendek18} originally proposed the parameter set
\begin{equation*}
\alpha=1,\quad \beta=1,\quad  \sigma =1,\quad  \xi=0.15,\quad \eta=0.000314, \quad \lambda=8.
\end{equation*}

However, this calibration turns out to be problematic from a numerical and modeling standpoint. In fact, it is readily observed that relatively large values of $\lambda$ -- in practice, values exceeding unity are already sufficient -- lead to pronounced instabilities in the model. In particular, the price process tends to exhibit explosive behavior regardless of the values assigned to the remaining parameters. Moreover, for large values of $\lambda$, the volatility factor $Y$ in \eqref{Platen1} rapidly converges to an asymptotic constant, effectively reducing the model to a Black-Scholes-type framework with constant volatility. Such a regime carries limited financial relevance, as it fails to reproduce key market stylized facts, including the implied volatility smile and skew.

For these reasons, in the present numerical study we depart from the parameter values adopted in \cite{platrendek18} and instead consider the following set, which allows us to work in a more stable environment\footnote{The purpose of this section is not to advocate the financial relevance of our parameter choice relative to that of \cite{platrendek18}, but to carry out a purely numerical investigation. Issues of estimation and calibration to real market data will be addressed in subsequent work.}:
\begin{equation*}
\alpha=1,\quad \beta=1,\quad  \sigma =1,\quad  \xi=0.05,\quad \eta=0.002, \quad \lambda=0.2,
\end{equation*}
together with the initial conditions $Y_0=0.1$ and we consider $S_0=100$ with a risk-free rate $r=3\%$.
We retain  $T=1$ with $500$ steps for the time discretization. For the spatial discretization, i.e.  for the functional quantization of the Brownian motion, we employ the same structural setup as in the previous section, with $d_N=20$ and 
\begin{equation*}
N_1=30,\quad N_2=14,\quad N_3=6,\quad N_4=3,\quad N_5=2,
\end{equation*}
corresponding to a total of $N=15,120$ quantized trajectories. In this case, we deliberately restrict ourselves to a more parsimonious representation of the (infinite-dimensional) set of Brownian paths. As discussed earlier, this choice inherently depends on the nature of the problem under consideration and on the characteristics of the underlying model. Given the flexibility and numerical stability exhibited by the model under the selected parameter configuration, the chosen value of $N$ represents a suitable compromise between computational efficiency and approximation accuracy for our purposes.

We begin by visualizing a selection of sample paths of the processes $S$ and $y$, while consistently keeping the Monte Carlo simulation as a reference benchmark for comparison.

In Figure \ref{fig:PlatenYS_FQ_and_MC}  we display $100$ trajectories  of the processes $Y$ and $S$ quantized by product functional quantization (left), and simulated by a Monte Carlo simulation procedure (right) using a simple Euler scheme. One can observe the smooth trajectories produced by the quantization approach in contrast with the more "chaotic" Monte Carlo paths; moreover, the dispersion appears to be comparable. Keep in mind that the quantized paths are not equally (or identically) weighted.

\begin{figure}[!htbp]
    \centering
\includegraphics[width=0.45\linewidth]{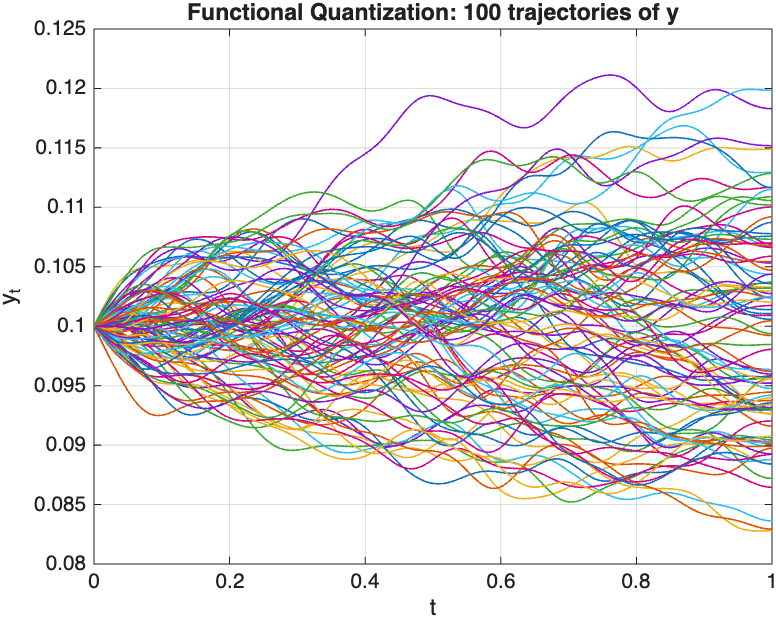} 
%   \includegraphics[width=0.48\textwidth]{YtrajectoriesLamb1.pdf}
   % \hfill
    \includegraphics[width=0.45\linewidth]{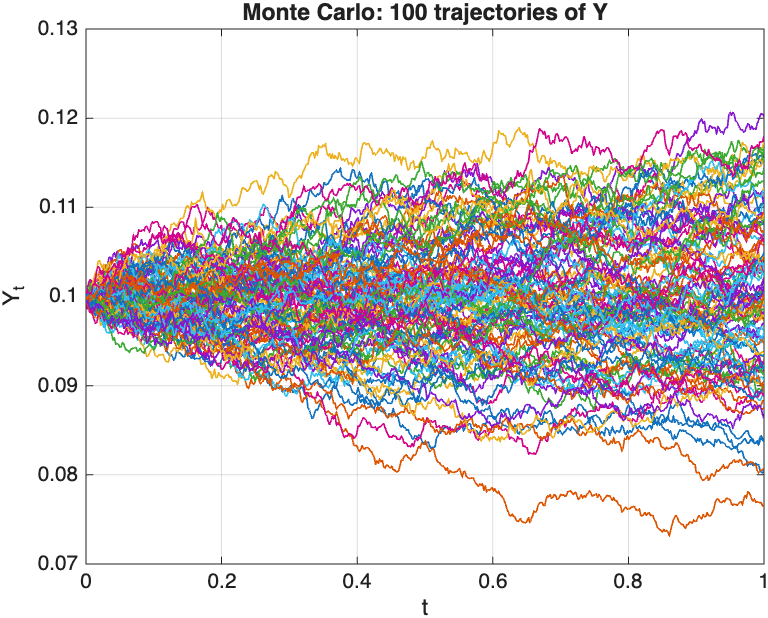}
%\hfill
\includegraphics[width=0.45\linewidth]{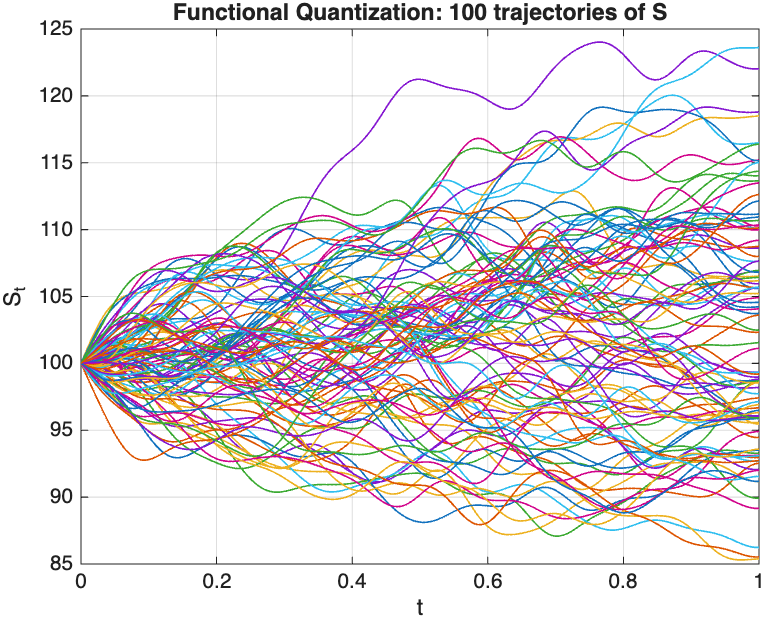} 
%   \includegraphics[width=0.48\textwidth]{YtrajectoriesLamb1.pdf}
   % \hfill
    \includegraphics[width=0.45\linewidth]{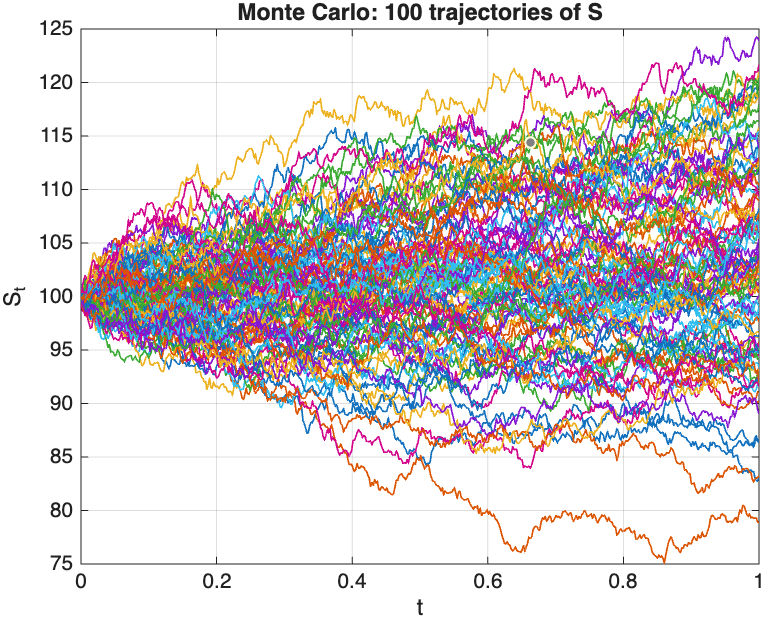}
    \caption{One hundred sample paths of the processes $y$ and $\hat S$ obtained via functional quantization (left panels) and those generated by Monte Carlo simulation with $10^5$ paths. The model parameters are $\alpha=\beta=  \sigma =1,   \xi=0.05,  \eta=0.002,  \lambda=0.2$.  $S_0=100;\ r=0;\ T=1$ with $500$ steps for the time discretization. }
    \label{fig:PlatenYS_FQ_and_MC}
\end{figure}

Before turning to the pricing results, we briefly comment on the impact of the model parameters on the behavior of the underlying asset.

The parameters $\alpha, \beta$ and $ \sigma$  enter the dynamics of $S$ only indirectly through the volatility factor $Y$, as a result, variations in these parameters have a limited impact on the fluctuations of the asset price. For the sake of brevity, we omit a detailed sensitivity analysis of the sample paths of the processes $S$ and $Y$ with respect to these model parameters, which is analogous to the study conducted in the previous section. A similar observation applies to $\xi$ and $\eta$: however, this holds only because we operate in a regime characterized by moderate values of $\lambda$ which emerges as the dominant parameter of the model. Indeed, $\lambda$ plays a central role both from an economic and financial perspective -- being associated with market activity and responsible for the memory effect in this non-Markovian framework -- and from a numerical standpoint, as it has a substantial impact on the trajectories of both $y$ and $S$.

In particular, Figure~\ref{fig:Platen_impact_Y_lambda} illustrates the effect of varying $\lambda$ on the trajectories of the process $y$, as obtained via functional quantization and Monte Carlo simulation. Consistent with the previous discussion, for high values of $\lambda$, each path of the process $Y$ becomes effectively flat, thereby reducing the model to a Black-Scholes-type  constant volatility setting. Such a regime is unable to reproduce the implied volatility smile and therefore lacks the flexibility required to capture key market features. On the contrary, with our choice $\lambda=0.2$, the model exhibits a markedly improved level of stability and, in particular, does not display significant distortions in the distribution of either the underlying asset or the volatility process for reasonable variations of $\xi, \eta, \alpha, \beta, \sigma$, as well as of the initial condition $Y_0$.

\begin{figure}[!htbp]
\centering

\begin{minipage}{0.48\linewidth}
\centering
\includegraphics[width=\linewidth]{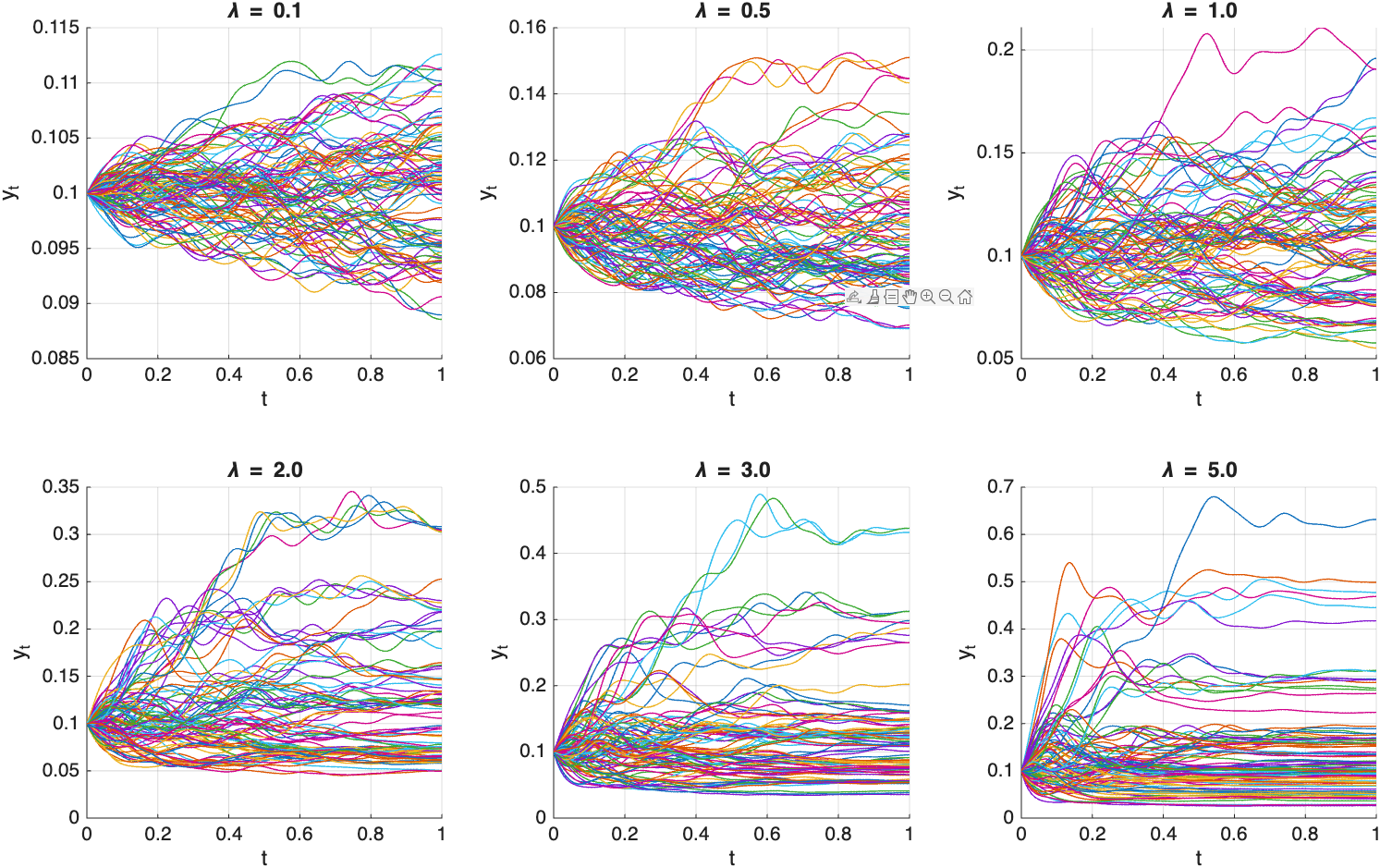}
\subcaption{Functional Quantization}
\end{minipage}
\hfill
\begin{minipage}{0.48\linewidth}
\centering
\includegraphics[width=\linewidth]{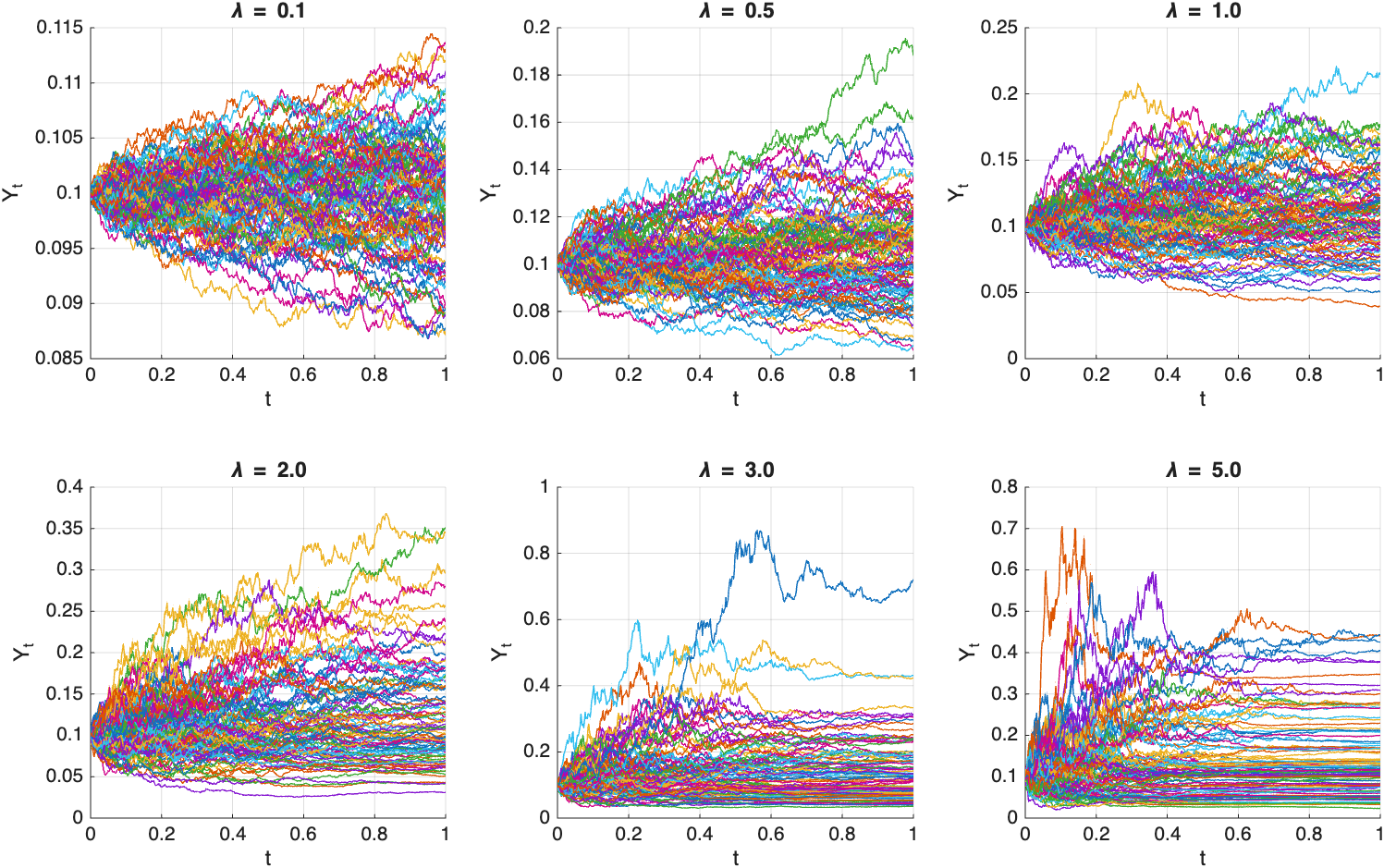}
\subcaption{Monte Carlo simulation}
\end{minipage}
    \caption{One hundred sample paths of the (codewords for the) process $Y$ corresponding to different values of $\lambda=0.1,\ 0.5,\ 1,\ 2,\ 3,\ 5$, with trajectories obtained via functional quantization (a)  and those generated by Monte Carlo simulation (b). }
    \label{fig:Platen_impact_Y_lambda}
\end{figure}

Finally, Figure~\ref{fig:Platen_impact_S_lambda} highlights the increasingly explosive behavior of the process $\hat S$ as the parameter $\lambda$ grows, ultimately leading to pronounced model instability.

\begin{figure}[!htbp]
\centering

\begin{minipage}{0.48\linewidth}
\centering
\includegraphics[width=\linewidth]{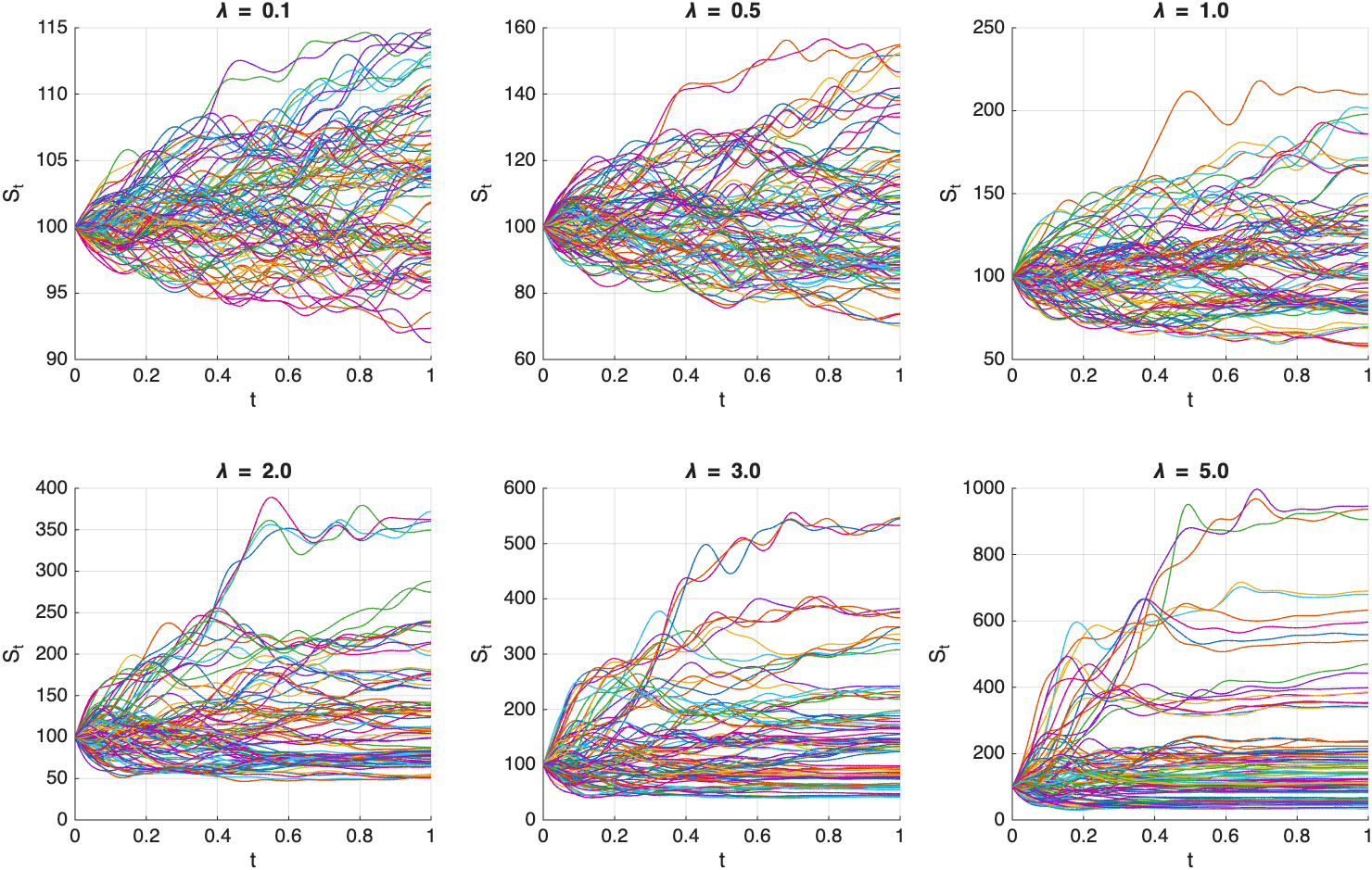}
\subcaption{Functional Quantization}
\end{minipage}
\hfill
\begin{minipage}{0.48\linewidth}
\centering
\includegraphics[width=\linewidth]{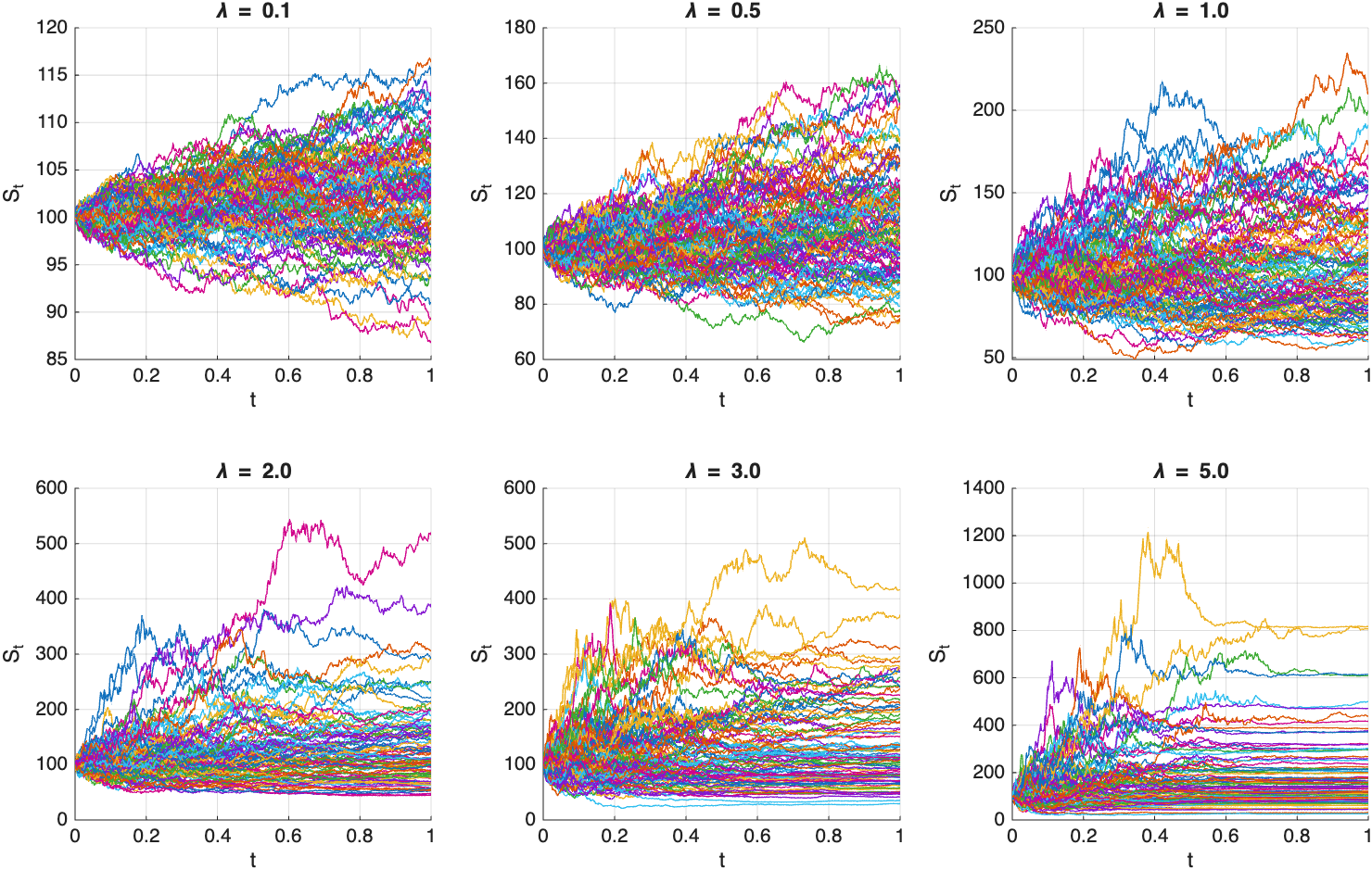}
\subcaption{Monte Carlo simulation}
\end{minipage}
    \caption{One hundred sample paths of the (codewords for the) process $\hat S$ corresponding to different values of $\lambda=0.1,\ 0.5,\ 1,\ 2,\ 3,\ 5$, with trajectories obtained via functional quantization (a) and those generated by Monte Carlo simulation (b). }
    \label{fig:Platen_impact_S_lambda}
\end{figure}

\subsubsection{Pricing of the Zero-Coupon Bond  in the Model of  \cite{platrendek18}}

We now consider the pricing of the zero-coupon bond. 
In Table \ref{TableZPB} we compare the price of a zero-coupon bond given by Monte Carlo (together with a confidence interval at $95\%$) with the price obtained by functional quantization denoted with $FQ$ for  different maturities $T=0.5,1,3,5,10,30,50$. 
Note that the functional quantization price is always (slightly) smaller that the Montecarlo ones: this is mostly likely due to the fact that for any stationary quantization $\hat S$ as the one we used, we have for every convex function $f$ that $\mathbb{E}[f(\hat S)]\leq \mathbb{E}[f( S)]$. In particular, it holds for $f(x)=1/x$, see e.g. \cite{zbMATH08132509}.

%

%==================== RESULTS ====================
%[FQ]  S0*E_w[1/S_T] (S full)  = 0.969000291879
%[FQ]  S0*E_w[1/S_T] (S naive) = 0.969083885958
%[MC]  S0*E[1/S_T]   (M=100000)    = 0.970886874249
 %     95% CI: [0.970360429586, 0.971413318911]
%=================================================

 \begin{table}[h!]
\begin{center}
\begin{tabular}{|c||c|c|c|c|}
\hline\hline
 T   &  FQ &     MC  &            $95\%$ C.I.         &        Discount Factor\\
 \hline\hline
    0.5 &   0.9830&   0.9849  &  [0.9845, 0.9853]     &       0.9851\\
\hline
      1  &  0.9673&    0.9699    &[0.9694, 0.9704]        &    0.9704\\ 
\hline
      3    &0.9096&   0.9135    &[0.9128, 0.9142]           &  0.9139\\  
\hline
      5    &0.8566&    0.8604   & [0.8596, 0.8612]         &  0.8607\\  
\hline
     10    &0.7373&     0.7405    &[0.7397, 0.7413]        &  0.7408\\ 
\hline
     20    &0.5466&    0.5487   & [0.5480, 0.5494]       &    0.5488\\  
\hline
     30    &0.4040&    0.4064    &[0.4058, 0.4070]      & 0.4065\\  
\hline
     50    &0.2200&    0.2231    &[0.2228, 0.2235]       &    0.2231\\ 
\hline\hline
 \end{tabular}
\end{center}
\caption{ Price of the Zero-Coupon bond for different maturities $T=0.5,1,3,5,10,30,50$ years given by  functional quantization (FQ) vs Monte Carlo simulation (MC) together with the corresponding $95\%$ confidence interval. On the last column we also display the discount factor, namely the risk neutral price of the ZCB. The parameters are as follows: 
 $\alpha=\beta= \sigma =1; \xi=0.05; \eta=0.002$ and $Y_0=0.1; S_0=100;r=3\%$. 
 For the functional quantization of the  Brownian motion we choose $d_N=20$  with $N_1=30$, $N_2 = 14$, $N_3=6$, $N_4 = 3$, $N_5=2$, corresponding to  $N=15,120$ trajectories. The size of  the Monte Carlo simulation is $10^5$. }
\label{TableZPB}\end{table}

Table \ref{TableZPB} shows that the functional quantization prices are in close agreement with their Monte Carlo counterparts for all maturities. Both estimates also coincide with the discount factor reported in the last column, which represents the risk-neutral price. 
We recall that deviations between benchmark prices under the real-world measure and risk-neutral prices may only occur if the Radon-Nikodym derivative associated with the change of measure becomes explosive, which does not seem to be the case for the parameter values considered in this numerical experiment.

To further confirm the necessity of introducing the correction term in the ODEs satisfied by the codewords, we report below several results on the sensitivity of the zero-coupon bond price with respect to the parameter $\lambda$. As shown in Table \ref{Table_lambda}, the price obtained via functional quantization (FQ)  is close to the Monte Carlo benchmark, whereas a fairly systematic bias emerges when the corrective term is omitted (the column ZCB FQ - $\hat S_{\textit{naive}}$). Moreover, the discrepancy becomes even more pronounced when both Stratonovich correction terms are removed, namely in the expression of $\hat S$  and in the ODEs governing the codewords $y_t$, as evidenced by the last column of the Table \ref{Table_lambda} (column ZCB FQ - $(y,\hat S)_{\textit{naive}}$). A similar effect can also be observed -- albeit to a lesser extent due to the lower price variability -- with respect to the other model parameters. We omit the corresponding results for the remaining parameters for the sake of brevity.

%lambda    S0_EinvS_FQ_full    S0_EinvS_MC        MC_95pct_CI        S0_EinvS_FQ_Snaive    S0_EinvS_FQ_naive_naive
   
   %  0.1          0.96953           0.97055      "[0.9702 , 0.9709]"         0.96976                  0.96976        
    % 0.3          0.96853           0.97077      "[0.9700 , 0.9715]"         0.96829                  0.96829        
    % 0.5          0.96789           0.97072      "[0.9696 , 0.9718]"         0.96626                  0.96626        
      % 1           0.9678           0.97167      "[0.9698 , 0.9735]"         0.95843                   0.9584        
     %  3          0.97423           0.97005      "[0.9666 , 0.9735]"         0.91408                  0.91169        
      % 5          0.97766           0.97129      "[0.9667 , 0.9758]"         0.87754                   0.8691        
      % 8          0.96393           0.96654      "[0.9607 , 0.9724]"         0.82977                  0.81025     

 \begin{table}[h!]
\begin{center}
\begin{tabular}{|c||c|c|c|c|c|}
\hline\hline
$\lambda$	 &  FQ &	 MC	&   CI $95\% $  &  FQ - $\hat S_{\textit{naive}}$ &	
  FQ - $(y, \hat S)_{\textit{naive}}$ \\
\hline \hline
 0.1  &        0.969  &         0.970&      (0.9708 , 0.9709)  &       0.969   &           0.969 \\       
   \hline
0.3  &        0.968 &           0.970&      (0.9700 , 0.9715)    &     0.968 &                 0.968\\        
   \hline
0.5  &        0.967  &         0.970&      (0.9696 , 0.9718)    &     0.966 &                 0.966\\       
   \hline
   1  &         0.967 &           0.971&      (0.9698 , 0.9735) &         0.958 &                   0.958\\        
\hline
 3  &        0.974  &        0.970 &      (0.9666 , 0.9735) &         0.914 &                  0.911\\        
     \hline
5     &     0.977 &          0.971 &      (0.9667 , 0.9758) &        0.877 &                   0.869\\        
      \hline
  8  &        0.963 &           0.966&       (0.9607 , 0.9724) &        0.829 &                  0.810\\ 
\hline\hline
 \end{tabular}
\end{center}
\caption{ Impact of the parameter $\lambda=0.1,0.3,0.5,1.3,5,8$  on the price of the Zero-Coupon Bond (ZCB) with maturity $T=1$ obtained with the product functional quantization (column  FQ), Monte Carlo simulation (column MC, together with the $95\%$ Confidence Interval), the naive quantization for the underlying, that is where the Stratonovich correction has been removed only from $\hat S$ (column  FQ - $\hat S_{\textit{naive}}$), and the completely  naive quantization, where the Stratonovich correction terms have been removed from both $\hat S$ and the ODEs \eqref{eq:codewordY} of the codewords $y$ (column FQ - $(y, \hat S)_{\textit{naive}}$). The other parameters are as follows: 
 $\alpha=\beta= \sigma =1; \xi=0.05; \eta=0.002$ and $Y_0=0.1; S_0=100;r=3\%$. 
 For the functional quantization of the  Brownian motion we choose $d_N=20$  with $N_1=30$, $N_2 = 14$, $N_3=6$, $N_4 = 3$, $N_5=2$, corresponding to  $N=15,120$ trajectories. The size of  the Monte Carlo simulation is $10^5$. }
\label{Table_lambda}\end{table}

In general, product functional quantization is not proposed to bypass the Monte Carlo method: instead, it should be thought of as a tool for variance reduction, by considering the possibility of using the quantization as a control variate. For example, in \cite{LehayReut2012}, a functional quantizer of Brownian motion is used as a control variate variable.
Ideas of this kind have also been discussed in \cite{corlaypages15} and \cite{LusPag2023}.

\subsubsection{Implied Volatility  in the Model of  \cite{platrendek18}}

In this subsection, we focus on the pricing of a call option on the GOP using the benchmark approach in the model of  \cite{platrendek18}. Using the general real-world pricing formula \eqref{benchmarkprice},  it follows that the price of a call on the GOP with maturity $T$ and strike $K$ is given by 
\begin{eqnarray}
Call_t(S_T)&=& S_t \mathbb E_t^{\mathbb{P}} \left[\frac{(S_T-K)^+}{S_T} \right].\label{callGOP}
\end{eqnarray}

In Figure~\ref{fig:imp_vol_FQ_MC_comparison_Platen}, we compare the implied volatility smiles obtained via functional quantization and Monte Carlo simulation across a range of maturities, from one month up to one year, for different values of the parameter $\lambda$. For moderate values of $\lambda$ (like the one used in this paper), the model is able to reproduce a volatility skew in line with the market, and the two methods produce consistent implied volatility smiles. As $\lambda$ increases, a visible discrepancy between the two approaches emerges; however, this difference is largely apparent rather than substantive and mainly due to numerical artifacts. Indeed, inspection of the vertical axis reveals that the smile shifts upward in level while becoming progressively flatter, so that the deviations between the two curves are in fact negligible. As recalled earlier, for sufficiently large values of $\lambda$ the model effectively reduces to a constant volatility Black-Scholes like framework, which is inherently unable to reproduce an implied volatility skew.

\begin{figure}[htbp]
    \centering
\includegraphics[width=1\linewidth]{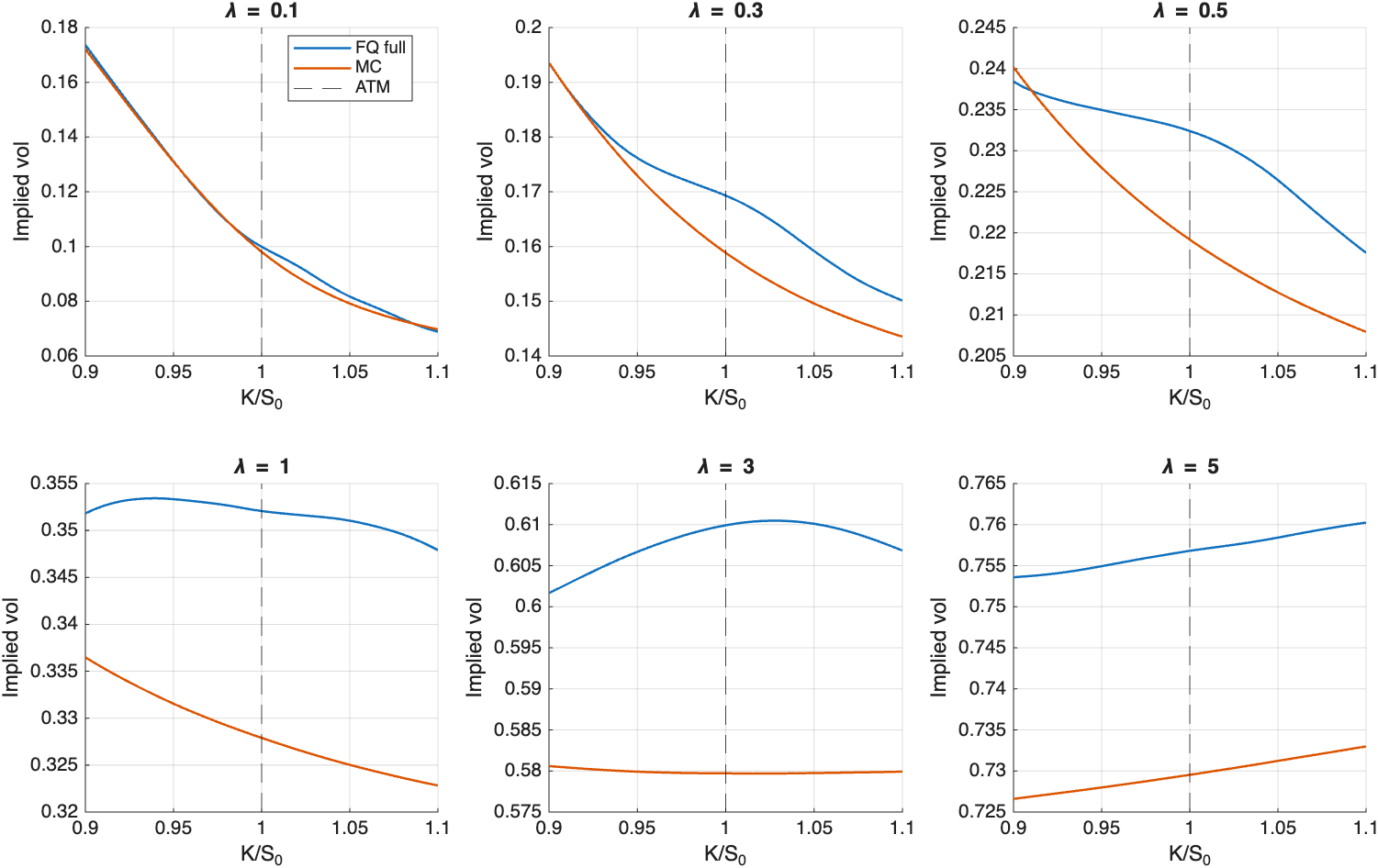} 
    \caption{Impact of the parameter $\lambda=0.1,0.3,0.5,1.3,5$  on the implied volatility smile with maturity $T=1$ obtained via functional quantization (FQ) and Monte Carlo simulation (MC).
    The other parameters are as follows: 
 $\alpha=\beta= \sigma =1; \xi=0.05; \eta=0.002$ and $Y_0=0.1; S_0=100;r=3\%$.}
    \label{fig:imp_vol_FQ_MC_comparison_Platen}
\end{figure}

Finally, Figure \ref{fig:imp_vol_FQ_FQnaive_comparison_Platen} illustrates the impact of the Stratonovich correction term on the implied volatility smile. For small values of $\lambda$, the smiles obtained under the different modeling assumptions are nearly indistinguishable, exhibiting substantial overlap. As $\lambda$ increases, however, the smile generated by a model that neglects the correction term in the dynamics of $S$ progressively deviates from that produced by the correctly implemented functional quantization scheme. A similar behavior is observed when both correction terms, in the dynamics of $S$ and of $y$, are omitted. Nonetheless, the same qualitative observation as before applies: as $\lambda$ grows, all implied volatility smiles become increasingly flat and tend to converge toward one another.

\begin{figure}[htbp]
    \centering
\includegraphics[width=1\linewidth]{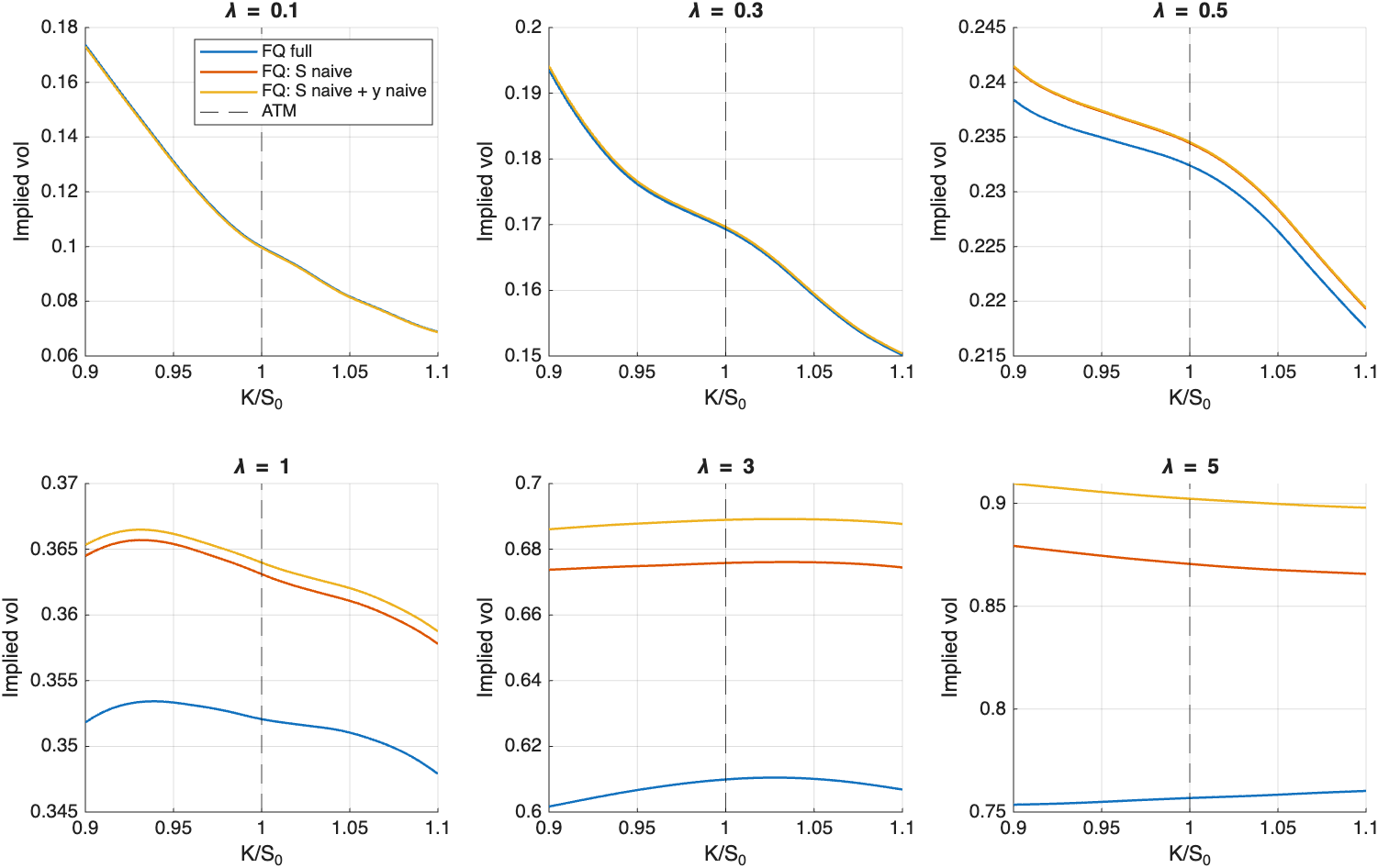} 
    \caption{Impact of the Stratonovich correction term as  $\lambda=0.1,0.3,0.5,1.3,5$ on the implied volatility smile at maturity $T=1$  generated with the functional quantization (FQ), the  and Monte Carlo simulation (MC).
    The other parameters are as follows: 
 $\alpha=\beta= \sigma =1; \xi=0.05; \eta=0.002$ and $Y_0=0.1; S_0=100;r=3\%$.Comparison of implied volatility smiles obtained via functional quantization (FQ), the naive  functional quantization FQ - $\hat S_{\textit{naive}}$ (FQ: S naive) and the completely  na\"ive  functional quantization FQ - $(y, \hat S)_{\textit{naive}}$ (FQ: S naive + y naive). The other parameters are as follows: 
 $\alpha=1;\beta=1; \sigma =0.1; \xi=0.05; \eta=0.002$ and $Y_0=0.1; S_0=100;r=3\%$. }
    \label{fig:imp_vol_FQ_FQnaive_comparison_Platen}
\end{figure}

\section{Conclusion}\label{section7}

In this paper, we have explored and extended several critical aspects of stochastic processes and their applications in financial modeling. 
We began by motivating the study of processes defined by Equation~\eqref{initialprocess}, demonstrating their relevance and applicability through illustrative examples. The introduction of product functional quantization, building upon the Karhunen-Lo\`{e}ve expansion for Brownian motion, proved to be a cornerstone of our approach. By extending the classical Lamperti transform method, we successfully addressed the challenges posed by the presence of memory terms in the diffusion coefficients, a key feature of the processes under investigation.
Our examination of the model proposed by \cite{GuyonVolMostlyPathDep2022} yielded valuable results, including a new approach to prove the existence and uniqueness of strong solutions for path-dependent SDEs. The application of the Lamperti transform in this context further illuminated the relationship between SDEs and ODEs in these complex systems.
The analysis of the \cite{platrendek18} model required additional care, due to the memory term in its diffusion coefficient. Our extended Lamperti transform proved instrumental in handling this complexity, demonstrating the versatility and power of our approach. The numerical illustrations,   provided for these models, offer practical insights into its behavior and potential applications.
Future research directions may include further generalizations of these processes, exploration of their applications in other domains beyond finance, and the development of more sophisticated numerical methods for their analysis, like e.g. the (multi-level) Romberg extrapolation, see \cite{pagesprintem05}. Additionally, the interplay between the theoretical foundations established here and empirical studies of financial markets could yield fruitful insights for both theory and practice.

\bibliography{biblio}
\bibliographystyle{apalike}
\end{document}